\documentclass[11pt, draftclsnofoot, onecolumn]{IEEEtran}

\usepackage{graphicx}
\usepackage{subfigure}
\usepackage{caption2}
\usepackage{amssymb}
\usepackage{amsthm}
\usepackage{amsmath}
\usepackage{cite}
\DeclareMathOperator{\Span}{span}

\DeclareMathOperator{\Rowrk}{rowrk}
\DeclareMathOperator{\snr}{SNR}
\begin{document}
%
\title{MISO Networks with Imperfect CSIT: A Topological Rate-Splitting Approach}

\author{Chenxi Hao and Bruno Clerckx 
\thanks{Chenxi Hao and Bruno Clerckx are with the Communication and Signal Processing group of Department of Electrical and Electronic Engineering, Imperial College London. Chenxi Hao is also with Beijing Samsung Telecom R\&D Center. Bruno Clerckx is also with the School of Electrical Engineering, Korea University.}
}

\maketitle

\begin{abstract}
Recently, the Degrees-of-Freedom (DoF) region of multiple-input-single-output (MISO) networks with imperfect channel state information at the transmitter (CSIT) has attracted significant attentions. An achievable scheme is known as rate-splitting (RS) that integrates common-message-multicasting and private-message-unicasting. In this paper, focusing on the general $K$-cell MISO IC where the CSIT of each interference link has an arbitrary quality of imperfectness, we firstly identify the DoF region achieved by RS. Secondly, we introduce a novel scheme, so called Topological RS (TRS), whose novelties compared to RS lie in a multi-layer structure and transmitting multiple common messages to be decoded by groups of users rather than all users. The design of TRS is motivated by a novel interpretation of the $K$-cell IC with imperfect CSIT as a weighted-sum of a series of partially connected networks. We show that the DoF region achieved by TRS covers that achieved by RS. Also, we find the maximal sum DoF achieved by TRS via hypergraph fractional packing, which yields the best sum DoF so far. Lastly, for a realistic scenario where each user is connected to three dominant transmitters, we identify the sufficient condition where TRS strictly outperforms conventional schemes.
\end{abstract}
\newtheorem{myprop}{Proposition}
\newtheorem{mycoro}{Corollary}
\newtheorem{myremark}{Remark}
\newtheorem{mydef}{Definition}
\newtheorem{mylemma}{Lemma}

\section{Introduction}\label{sec:intro}
Channel state information at the transmitter (CSIT) is crucial to the downlink multi-user transmission strategies. However, acquiring accurate CSIT is challenging in practical systems. In wireless systems like LTE, the CSIT is obtained by uplink-downlink reciprocity in Time Division Duplex setup, or by user feedback in Frequency Division Duplex Setup. In multi-cell scenarios, the CSIT has to be shared among the transmitters in order to perform coordinated beamforming and/or joint transmission. Those procedures result in imperfect CSIT due to the channel estimation error, quantization error and the Doppler effect caused by the latency in the feedback link and backhaul link. Performing interference mitigation techniques designed for perfect CSIT using imperfect CSIT results in undesirable multi-user interferences, which deteriorate the system performance. Hence, the fundamental question that should therefore be addressed is how to design proper transmission strategies for the imperfect CSIT setting.

Recent work \cite{Ges12} found the optimal DoF region of a two-user multiple-input-single-output (MISO) broadcast channel (BC) with a mixture of perfect delayed CSIT and imperfect instantaneous CSIT. However, one corner point of the optimal DoF region is achieved by a rate-splitting (RS) approach which does not rely on delayed CSIT and is applicable to the scenario with only imperfect instantaneous CSIT. Reminiscent to Han-Kobayashi scheme \cite{HanKobayashi,Etkin08}, each user's message in RS is split into a common and a private part. The private messages are unicast to their respective intended users along Zero-Forcing (ZF) precoders using a fraction of the total power. The common messages are encoded into a super common message, and the super common message is multicast using the remaining power. At the receiver side, each user firstly decodes the super common message and proceeds to decode the desired private message afterwards using successive interference cancellation (SIC). This RS approach can be easily applied to the $K$-user MISO BC. Considering that the CSIT error of user $k$ decays with signal-to-noise-ratio (SNR) as $\snr^{-\alpha_k}$ where $0{\leq}\alpha_1{\leq}{,}\cdots{,}{\leq}\alpha_K{\leq}1$ is commonly termed as the CSIT qualities, the sum DoF achieved by RS is $1{+}\sum_{k{=}1}^{K{-}1}\alpha_k$. Based on the assumption of real input and channel vectors, the optimality of this result was shown in \cite{Davoodi14}.

Since then, there have been extensive researches on RS. The sum rate analysis in the presence of quantized CSIT and the precoder optimization for sum-rate maximization were investigated in \cite{RateAnalysis,HamdiRS}, respectively. Literature \cite{Mingbo_RS_mMIMO} extended the idea of RS into the massive Multiple-Input-Multiple-Output (MIMO) deployment and proposed a Hierarchical RS which exploits the spatial correlation matrices to effectively tackle the multiuser interferences. Other related works on MISO BC can be found in \cite{Tandon12,icc13freq,pimrc2013,Jinyuan_evolving_misobc,Borzoo_Kuser}. The application of RS to the two-cell MISO interference channel (IC) was firstly reported in \cite{xinping_miso_ic}. The scheme was later on extended to the two-cell MIMO IC with asymmetric number of antennas in \cite{Hao15MIMODoF}.

However, designing a scheme suitable for the $K$-cell IC is a non-trivial step, because the interference overheard by a single user come through $K{-}1$ different links and the CSIT of each link may have a particular quality of imperfectness. A promising idea can be drawn from the HRS designed under massive MIMO setting \cite{Mingbo_RS_mMIMO}. In HRS, users are clustered based on the similarity of their transmit correlation matrices. Then, the users in different groups are separated by statistical Zero-Forcing beamforming (ZFBF) using long term CSIT, while the users in the same group are separated by ZFBF using instantaneous CSIT. Due to the imperfect grouping and imperfect instantaneous CSIT, there exists residual intra- and inter-group interference that impacts the system performance. To deal with this problem, RS is evolved to HRS by integrating an outer RS and an inner RS. The outer RS tackles the inter-group interference by multicasting a system common message to be decoded by all users, while the inner RS tackles the intra-group interference by transmitting a group common message for each group. Using SIC, each user decodes the system common message, the group common message of the corresponding group and the desired private message sequentially.

A similar problem occurs in the $K$-cell IC if the users can be categorized into groups such that there are identical intra-group CSIT qualities, and the intra-group CSIT quality is smaller than the inter-group CSIT quanlities. Then, the users belonging to the same and different groups are separated by ZFBF using intra- and inter-group CSIT, respectively. The residual inter- and intra-group interference is tackled by the outer- and inner-layer RS, respectively. Although such a user-grouping method is only applicable to a very limited class of CSIT quality topologies, the concepts of transmitting group common messages and multi-layer structure shed light on the essential point of establishing the transmission block for the general $K$-cell IC with arbitrary CSIT quality topology. The main contributions are stated as follows.

\subsubsection{Achievable DoF region of RS}
Focusing on the $K$-cell MISO IC where the CSIT of each interference link has an arbitrary quality of imperfectness, we firstly consider a logical extension of the RS designed for two-cell MISO IC. Each transmitter divides the message intended for the corresponding user into a common and a private part. Each private message is unicast using an arbitrary fraction of the total power, while the remaining power at each transmitter is employed to multicast the common message to be decoded by all users. We characterize the resultant DoF region and show that it covers the DoF region achieved by conventional ZFBF (private message transmission only) with power control.
\subsubsection{Topological RS with weighted-sum interpretation}
We propose a novel scheme so called Topological RS (TRS), that is suitable for the general $K$-cell MISO IC with arbitrary CSIT quality topology. Unlike RS, each user' message in TRS is split into $N$ parts, i.e., $\mathcal{W}_k{\triangleq}\{w_k^1{,}w_k^2{,}\cdots{,}w_k^N\}$, where $w_k^1$ is a private message to be decoded by user $k$, while $w_k^i{,}i{\geq}2$, is a common message to be decoded by a group of users $\mathcal{R}_k^i$. The power allocated to the common messages and the user group $\mathcal{R}_k^i$ are determined based on the specific CSIT quality topology, so that the group common message $w_k^i$ is drowned into the noise at other users via ZFBF. We show that the DoF region achieved by TRS covers that achieved by RS.

The TRS scheme is inspired by a novel interpretation of the $K$-cell MISO IC with imperfect CSIT as a weighted-sum of a series of partially connected networks with different topologies. The weights of the partially connected networks stand for their separations in the power domain. This weighted-sum interpretation explicitly shows whether or not a user is interfered with one another, thus helping us generating group common messages. Moreover, the DoF region achieved by TRS is interpreted as a weighted-sum of that achieved in those networks, thus allowing us to employ methodologies applicable for partially connected networks to analyze the DoF region achieved with imperfect CSIT.

\subsubsection{Sum DoF using graph theory tools}
As a consequence of the weighted-sum interpretation, studying the sum DoF achieved by TRS is equivalent to studying the sum DoF in each obtained partially connected network. Then, for each partially connected network, we propose two common message groupcasting methods from a graph theory perspective. These two methods called \emph{orthogonal groupcasting} and \emph{maximal groupcasting}\footnote{When a common message is to be decoded by a subset of all users, it is referred as a common message groupcasting.} are respectively built upon the packing and fractional packing of the hypergraph defined by the network topology. The maximal groupcasting method yields the maximal sum DoF in each partially connected network, thus giving the maximal sum DoF achieved by TRS. This sum DoF result is no less than that achieved by RS and ZFBF with power control.

\subsubsection{Results in realistic scenarios}
As it has been shown that in many practical deployments each user has two dominant interferers \cite{BrunoBook}, we consider a realistic setting where each user is connected to its closest three transmitters. We design TRS for a class of CSIT quality topology, which is featured by that the two incoming interference links associated with each user have unequal CSIT qualities $a$ and $b$ where $0{\leq}a{\leq}b{\leq}1$. With maximal groupcasting, we characterize the achievable sum DoF by TRS and show that it is within the range $\left[\frac{K}{3}(1{+}\frac{b}{2}{+}\frac{3a}{2}){,}\frac{K}{3}(1{+}b{+}a)\right]$. For a cyclic CSIT quality topology, we find that the proposed TRS approach strictly outperforms ZFBF with power control as long as $b{+}3a{>}\frac{6}{K}\lfloor\frac{K}{2}\rfloor{-}2$, where $\lfloor\frac{K}{2}\rfloor$ is the maximum integer that is not greater than $\frac{K}{2}$.

The rest of the paper is organized as follows. The system model is introduced in Section \ref{sec:SM}. In Section \ref{sec:pre}, we revisit ZFBF with power control and characterize the DoF region achieved by RS with common message multicasting. In Section \ref{sec:TRS}, we propose the generalized framework of TRS approach together with its weighted-sum interpretation, and study its achievable DoF region and sum DoF performance. Section \ref{sec:real} studies the sum DoF achieved by TRS in realistic scenarios. Section \ref{sec:conclusion} concludes the paper.

Notations: Bold upper and lower letters denote matrices and vectors respectively. A symbol not in bold font denotes a scalar. $({\cdot})^H$, $({\cdot})^T$ and $({\cdot})^\bot$ respectively denote the Hermitian, transpose and the null space of a matrix or vector. $\parallel{\cdot}\parallel$ refers to the norm of a vector. $\Rowrk(\mathbf{A})$ stands for the row rank of matrix $\mathbf{A}$, while $\Span(\mathbf{A})$ refers to the subspace spanned by $\mathbf{A}$. The term $\mathbf{1}_M$ refers to a $M{\times}1$ vector with all $1$ entries. For a set $\mathcal{A}$, $|\mathcal{A}|$ represents its cardinality; for a complex number $a$, $|a|$ stands for its absolute value. The term $1_{C}$ is the indicator function, it is equal to $1$ if condition $C$ holds; otherwise, it is equal to $0$. $\mathbb{E}\left[{\cdot}\right]$ refers to the statistical expectation. $(a)^+$ stands for $\max(a{,}0)$. ${a}\bmod{n}$ calculates the modulus of integer $a$ with the respect of integer $n$. $\lfloor{a}\rfloor$ refers to the maximal integer that is no greater than $a$. 

\section{System Model}\label{sec:SM}
\subsection{$K$-cell Interference Channel}\label{sec:KcellIC}
In this paper, we consider a $K$-cell interference channel, where each transmitter is serving one user in each cell. We assume that there is a sufficient number of antennas, i.e., $K$, at each transmitter in order to perform interference nulling strategies, such as ZFBF etc, while there is a single antenna at each user. The signal transmitted by a certain transmitter is denoted by $\mathbf{s}_k{\in}\mathbb{C}^{K{\times}1}{,}{\forall}k{\in}\mathcal{K}$ where $\mathcal{K}{\triangleq}\{1{,}\cdots{,}K\}$, and it is subject to the power constraint $P$. Then, the received signals write as
\begin{IEEEeqnarray}{rcl}
y_k&{=}&\sum_{j{=}1}^{K}g_{kj}\mathbf{h}_{kj}^H\mathbf{s}_j{+}n_k{,}{\forall}k{\in}\mathcal{K}{,}
\end{IEEEeqnarray}
where $n_k$ is the additive white Gaussian noise with zero mean and unit variance; $\mathbf{h}_{kj}{\in}\mathbb{C}^{K{\times}1}$ represents the channel between transmitter $j$ and user $k$, whose entries are i.i.d Gaussian with zero mean and unit variance; $g_{kj}{\in}\{0{,}1\}$, ${\forall}k{,}j{\in}\mathcal{K}$, is a binary variable. When $g_{kj}{=}1$, it means that transmitter $j$ is connected to user $k$. When $g_{kj}{=}0$, it means that the signal sent out by transmitter $j$ is drowned into the noise at user $k$ due to the path loss. For convenience, let us use $\mathcal{G}{\triangleq}\{g_{kj}\}_{{\forall}k{,}j{\in}\mathcal{K}}$ to denote the network topology.

Throughout the paper, we consider $g_{kk}{=}1$, ${\forall}k{\in}\mathcal{K}$, and thus $P$ is referred as the SNR. For the interference links, we consider that
\begin{itemize}
  \item in Section \ref{sec:pre} and \ref{sec:TRS}, we have $g_{kj}{=}1$, ${\forall}k{\in}\mathcal{K}$ and ${\forall}j{\in}\mathcal{K}{\setminus}j$. This indicates a fully connected network where the interference-to-noise-ratio (INR) is equal to SNR;
  \item in Section \ref{sec:real}, if $j{=}(k{-}1){\bmod}K{,}(k{+}1){\bmod}K$, we have $g_{kj}{=}1$; otherwise $g_{kj}{=}0$. This corresponds to a homogeneous cellular network where user $k$ is only connected to three dominant transmitters \cite{BrunoBook}, i.e., transmitter $k$, $k{-}1$ and $k{+}1$. Note that a cyclic setting is assumed such that user $1$ is connected to transmitter $K$, $1$ and $2$, while user $K$ is connected to transmitter $K{-}1$, $K$ and $1$.
\end{itemize}

\subsection{CSIT Quality Topology}
We consider that the channel vector is expressed as $\mathbf{h}_{kj}{=}\hat{\mathbf{h}}_{kj}{+}\tilde{\mathbf{h}}_{kj}$, where $\hat{\mathbf{h}}_{kj}$ is the imperfect CSIT and $\tilde{\mathbf{h}}_{kj}$ represents the CSIT error, drawn from a continuous distribution.

For the link with $g_{kj}{=}1$, following the classical model firstly introduced in \cite{Ges12,Gou12}, we define the CSIT quality as
\begin{IEEEeqnarray}{rcl}
a_{kj}&{\triangleq}&
-\lim_{P{\to}\infty}\frac{\log_2\mathbb{E}\left[|\mathbf{h}_{kj}^H\hat{\mathbf{h}}_{kj}^\bot|^2\right]}{\log_2P}{,}
{\forall}k{\in}\mathcal{K}{,}{\forall}j{\in}\mathcal{K}{\setminus}k{,}g_{kj}{=}1{,}\label{eq:akj}
\end{IEEEeqnarray}
where the quantity $\mathbb{E}\left[|\mathbf{h}_{kj}^H\hat{\mathbf{h}}_{kj}^\bot|^2\right]$ represents the strength of the residual interference resulted by ZFBF using imperfect CSIT. The expectation is taken over both the imperfect CSIT $\hat{\mathbf{h}}_{kj}$ and the channel vector $\mathbf{h}_{kj}$. This expression is equivalent to $\mathbb{E}\left[|\mathbf{h}_{kj}^H\hat{\mathbf{h}}_{kj}^\bot|^2\right]{=}P^{-a_{kj}}{+}o(P^{-a_{kj}})$ when $P{\to}\infty$. This quantity implies that if transmitter $j$ unicasts a ZF-precoded private message using power $P^{a_{kj}}$, then the residual interference at user $k$ is drowned into the noise. From a DoF perspective, when $a_{kj}{\geq}1$, it is equivalent to having perfect CSIT because the interference can be forced within the noise level and the full DoF $K$ can be achieved by ZFBF \cite{Ges12,Gou12}; when $a_{kj}{=}0$, it is equivalent to the case without CSIT \cite{Ges12,Gou12}, because the interference term is received with the same power level as the desired signal and the resultant sum DoF is $1$. Hence, in this paper, we only focus on the case $0{\leq}a_{kj}{\leq}1{,}{\forall}k{\in}\mathcal{K}{,}{\forall}j{\in}\mathcal{K}{\setminus}k{,}g_{kj}{=}1$.

However, in Section \ref{sec:real}, the CSIT quality of the link with $g_{kj}{=}0$ is not defined, as the strength of the signal sent by transmitter $j$ is drowned into the noise at user $k$ even without performing ZFBF.

Moreover, we consider that the CSIT qualities vary across the links. This leads to a CSIT topology defined by $\mathcal{A}{\triangleq}\{a_{kj}\}_{{\forall}k{\in}\mathcal{K}{,}{\forall}j{\in}\mathcal{K}{\setminus}k{,}g_{kj}{=}1}$. Note that the CSIT qualities of the direct links $a_{kk}$, ${\forall}{\in}\mathcal{K}$, are not included because their values only offer beamforming gain, which does not make a difference on the DoF performance. A CSIT quality topology $\mathcal{A}$ can be also defined using a table (see the fully connected IC in Figure \ref{fig:3user_het} for example), where each row stands for the CSIT qualities of the incoming links of a certain user, while each column represents the CSIT qualities of the outgoing links of a certain transmitter.

\subsection{Rate Splitting}
The message of each user is assumed to be split into $N$ parts, i.e., $\mathcal{W}_k{\triangleq}\{w_k^1{,}w_k^2{,}\cdots{,}w_k^N\}$, where $w_k^1$ is the private message to be decoded by user $k$ only, while $w_k^i{,}i{\geq}2$ is a common message to be decoded by a group of users $\mathcal{R}_k^i$. We consider that each transmitter only has the message intended for its corresponding user.  With imperfect local CSIT, the knowledge of network topology $\mathcal{G}$, and the CSIT quality topology $\mathcal{A}$, the encoding function for each transmitter can be expressed as
\begin{IEEEeqnarray}{rcl}
\mathbf{s}_k&{=}&f(\mathcal{W}_k{,}\hat{\mathbf{h}}_{kk}{,}\{\hat{\mathbf{h}}_{kj}\}_{{\forall}j{\in}\mathcal{K}{\setminus}k}{,} \mathcal{G}{,}\mathcal{A}){,}{\forall}k{\in}\mathcal{K}.
\end{IEEEeqnarray}

At the receiver side, we consider that there is \emph{perfect local CSIR}, namely user $k$ perfectly knows the effective channels, i.e., the multiplication of the precoders and the channel vectors, so as to decode the desired signal. Let $R_k^i$ denote the rate of message $w_k^i$. A rate tuple $\left(\{R_k^1\}_{k{\in}\mathcal{K}}{,}\cdots{,}\{R_k^N\}_{k{\in}\mathcal{K}}\right)$ is said achievable if private message $w_k^1$ is decoded by user $k$, and common message $w_k^i{,}i{\geq}2$ is decoded by the group of users $\mathcal{R}_k^i$, with arbitrary small error probability. Then, the achievable DoF of a certain message $w_k^i$ is defined as $d_k^i{\triangleq}\lim\limits_{P{\to}\infty}\frac{R_k^i}{{\log}_2P}$. The achievable DoF of user $k$ is computed by $d_k{=}\sum_{i{=}1}^{N}d_k^i$.

Throughout the paper, the terminology \emph{common message groupcasting} means that a common message $w_k^i$ is to be decoded by a group of users $\mathcal{R}_k^i$. When the group contains all users, i.e., $\mathcal{R}_k^i{=}\mathcal{K}$, the common message groupcasting becomes \emph{common message multicasting}. When the group is formed by only one user, i.e., $\mathcal{R}_k^i{=}\{k\}$, it actually refers to a \emph{private message unicasting}. 

\section{Preliminaries}\label{sec:pre}
In this section, focusing on a fully connected network with equal SNR and INR, we revisit two benchmark schemes, i.e., conventional ZFBF with power control and RS approach with common message multicasting. For RS, we also propose its achievable DoF region in the fully connected $K$-cell MISO IC with imperfect CSIT.

\subsection{ZFBF with power control}
In conventional ZFBF with power control, transmitter $k$ delivers a private message $w_k$ to the corresponding user using power $P^{r_k}$, $r_k{\leq}1$, along a ZF-precoder $\mathbf{p}_k{\subseteq}\Span(\{\hat{\mathbf{h}}_{jk}^\bot\}_{{\forall}j{\in}\mathcal{K}{\setminus}k})$. The signal received by user $k$ can be expressed as
\begin{IEEEeqnarray}{rcl}
y_k&{=}&\underbrace{\mathbf{h}_{kk}^H\mathbf{p}_kw_k}_{P^{r_k}}{+}
\sum_{j{\in}\mathcal{K}{\setminus}k}\underbrace{\mathbf{h}_{kj}^H\mathbf{p}_jw_j}_{P^{r_j{-}a_{kj}}}{+} \underbrace{n_k}_{P^0}{.}\label{eq:yk_zfbf}
\end{IEEEeqnarray}
By treating the undesired private message as noise, the DoF achieved by each private message writes as
\begin{IEEEeqnarray}{rcl}
d_k&{\leq}&\left(r_k{-}\max_{j{:}j{\in}\mathcal{K}{\setminus}k}(r_j{-}a_{kj})^+\right)^+{,}{\forall}k{\in}\mathcal{K}{.}\label{eq:dk_zfbf}
\end{IEEEeqnarray}
This expression specifies the DoF region achieved by ZFBF with power allocation policy $\mathbf{r}{\triangleq}(r_1{,}\cdots{,}r_K)$. The DoF region achieved by ZFBF with power control, denoted by $\mathcal{D}_{ZF}$, is the union of DoF regions achieved with all the possible power allocation $\mathbf{r}$ where $r_k{\leq}1{,}{\forall}k{\in}\mathcal{K}$.

Notably, by performing ZFBF, the expression in \eqref{eq:yk_zfbf} can be regarded as the received signal in an IC where the direct links have unit gain, while the strength of the interference link is $P^{-a_{kj}}$, ${\forall}k{\neq}j$. Hence, a concise expression of $\mathcal{D}_{ZF}$ by eliminating the variables $\mathbf{r}$ can be obtained using \cite[Theorem 5]{Geng15TIN}.

\subsection{Rate-Splitting with common message multicasting}
The RS approach was firstly introduced focusing on a $2$-cell MISO IC with a symmetric CSIT setting, i.e., $a_{12}{=}a_{21}{=}a$. In \cite{xinping_miso_ic}, one user's message is split into a common and a private part, while the other user's message has a private part only. By unicasting the private messages along ZF-precoders using power $P^a$, and multicasting the common message using the remaining power $P{-}P^a$, the sum DoF $1{+}2a$ is achievable. This result is optimal for the $2$-cell MISO IC as it is identical to the optimal sum DoF of a two-user MISO BC with symmetric CSIT quality $a$.

The beauty of RS lies in forcing the residual interference caused by ZFBF with imperfect CSIT to the very weak interference regime, while introducing a strong interference, i.e., the common message, which is decodable by treating the private messages as noise. However, the achievability of RS in the general $K$-cell MISO IC remains an open problem. Here, we propose a logical extension of RS to the $K$-cell MISO IC. We consider that a certain group $\mathcal{S}{\subseteq}\mathcal{K}$ of users are active, while the remaining users are made silent. This assumption allows us to obtain an achievable DoF region by taking the union of all the possible subsets $\mathcal{S}{\subseteq}\mathcal{K}$ of users.

We consider a general RS approach where each active user's message is split into a private part $w_k^p$ and a common part $w_k^c$, ${\forall}k{\in}\mathcal{S}$. These two messages are transmitted using power $P^{r_k}$ and $P{-}P^{r_k}$, respectively, where $r_k{\leq}1$. The common messages $\{w_k^c\}_{k{\in}\mathcal{S}}$ are to be decoded by all the active users. The transmitted signal and received signal are expressed as
\begin{IEEEeqnarray}{rcl}
\mathbf{s}_k&{=}&\underbrace{\mathbf{p}_k^cw_k^c}_{P{-}P^{r_k}}{+}\underbrace{\mathbf{p}_k^pw_k^p}_{P^{r_k}}{,}{\forall}k{\in}\mathcal{S}{,}\\
y_k&{=}&\sum_{{\forall}j{\in}\mathcal{S}}\underbrace{\mathbf{h}_{kj}^H\mathbf{p}_j^cw_j^c}_{P}{+}\underbrace{\mathbf{p}_k^pw_k^p}_{P^{r_k}}{+}
\sum_{{\forall}j{\in}\mathcal{S}{\setminus}k}\underbrace{\mathbf{h}_{kj}^H\mathbf{p}_j^pw_j^p}_{P^{r_j{-}a_{kj}}}{+}
\underbrace{n_k}_{P^0}{,}\label{eq:ykRS}
\end{IEEEeqnarray}
respectively, where $\mathbf{p}_k^p{\subseteq}\Span(\{\hat{\mathbf{h}}_{jk}^\bot\}_{{\forall}j{\in}\mathcal{S}{\setminus}k})$ are ZF-precoders, while $\mathbf{p}_k^c$ are random precoders.

Each user firstly decodes all the common messages, and secondly recovers the desired private message after removing the common messages using SIC. Then, the DoF tuple achieved by the private messages and the common messages, denoted by $(d_1^p{,}{\cdots}d_K^p)$ and $(d_1^c{,}{\cdots}d_K^c)$ respectively, are such that
\begin{IEEEeqnarray}{lcl}
\sum_{k{\in}\mathcal{S}}d_k^c{\leq}1{-}\max_{j{\in}\mathcal{S}}r_j{,}{\quad} d_k^p{\leq}\left(r_k{-}\max_{j{:}j{\in}\mathcal{S}{\setminus}k}(r_j{-}a_{kj})^+\right)^+{,} {\forall}k{\in}\mathcal{S}{;} \quad d_j^c{=}d_j^p{=}0{,}{\forall}j{\in}\mathcal{K}{\setminus}\mathcal{S}{.}\label{eq:DcpRS}
\end{IEEEeqnarray}
The achievable DoF region by RS with active user set $\mathcal{S}$ and power allocation policy $\mathbf{r}$, denoted by $\mathcal{D}_{RS}(\mathcal{S}{,}\mathbf{r})$, is the set of all DoF tuple $(d_1{,}\cdots{,}d_K){=}(d_1^c{,}\cdots{,}d_K^c){+}(d_1^p{,}\cdots{,}d_K^p)$, for which \eqref{eq:DcpRS} holds.

Then, the DoF region achieved by RS is resulted by the union of the DoF regions achieved with all possible subsets $\mathcal{S}$ and power allocation policy $\mathbf{r}$, i.e., $\mathcal{D}_{RS}{\triangleq}\bigcup_{{\forall}\mathcal{S}{\subseteq}\mathcal{K}{,} {\forall}\mathbf{r}} \mathcal{D}_{RS}(\mathcal{S}{,}\mathbf{r})$. The following proposition settles $\mathcal{D}_{RS}$.
\begin{myprop}\label{prop:DRS}
In a fully connected $K$-cell MISO IC with equal SNR and INR and with CSIT quality topology $\mathcal{A}$, the DoF region achieved by RS with common message multicasting is
\begin{IEEEeqnarray}{rcl}
\mathcal{D}_{RS}&{=}&\bigcup_{{\forall}\mathcal{U}{\subseteq}\mathcal{K}}\mathcal{D}_{RS}(\mathcal{U})
\end{IEEEeqnarray}
where $\mathcal{D}_{RS}(\mathcal{U})$ is the set of $(d_1{,}\cdots{,}d_K){=}(d_1^c{,}\cdots{,}d_K^c){+}(d_1^p{,}\cdots{,}d_K^p)$ such that
\begin{IEEEeqnarray}{ccl}
d_k^p{=}0{,}{\forall}k{\in}\mathcal{K}{\setminus}\mathcal{U}{;}\,
0{\leq}d_k^p{\leq}1{,}{\forall}k{\in}\mathcal{U}{;}\,\sum_{l{=}1}^{m}d_{i_l}^p{\leq}\sum_{l{=}1}^{m}a_{i_{l{-}1}i_l}{,}
{\forall}(i_1{,}\cdots{,}i_m){\in}\Pi_{\mathcal{U}}{;}\label{eq:DRS_KUp}\\
0{\leq}d_k^c{\leq}1{,}{\forall}k{\in}\mathcal{K}{;}
0{\leq}d_k^p{+}\sum_{j{\in}\mathcal{K}}d_j^c{\leq}1{,} {\forall}k{\in}\mathcal{U}{;}\,
\sum_{j{\in}\mathcal{S}}d_j^c{+}\sum_{l{=}1}^{m}d_{i_l}^p{\leq}1{+}\sum_{l{=}2}^{m}a_{i_{l{-}1}i_l}{,} {\forall}(i_1{,}\cdots{,}i_m){\in}\Pi_{\mathcal{U}}{,}\label{eq:DRS_KUpc}
\end{IEEEeqnarray}
and $\Pi_{\mathcal{U}}$ is the set of all possible cyclic sequences\footnote{A cyclic sequence is a cyclically ordered subset of user indices without repetitions \cite{Geng15TIN}. For a certain subset $(i_1{,}\cdots{,}i_m)$, there are $(m{-}1)!$ distinct cyclic orders. For a user set $\mathcal{U}$, there exist $\sum_{m{=}2}^{|\mathcal{U}|}{{|\mathcal{U}|} \choose {m}}$ different subset $(i_1{,}\cdots{,}i_m)$ with $m{\geq}2$. Hence, $\Pi_{\mathcal{U}}$ have $\sum_{m{=}2}^{|\mathcal{U}|}{|\mathcal{U}| \choose {m}}(m{-}1)!$ cyclic sequences. For instance, let $\mathcal{U}{=}\{1{,}2{,}3\}$, then $\Pi_{\mathcal{U}}{=}\{1{,}2\}{,}\{1{,}3\}{,}\{2{,}3\}{,}\{1{,}2{,}3\}{,}\{1{,}3{,}2\}$.} of all subsets of $\mathcal{U}$ with cardinality no less than $2$.
\end{myprop}
\begin{proof}
  see Appendix A.
\end{proof}

Note that the DoF region $\mathcal{D}_{RS}(\mathcal{U})$ is obtained by scheduling all the users and choosing the following power allocation policy
\begin{IEEEeqnarray}{rcl}
r_k{\leq}0{,}k{\in}\mathcal{K}{\setminus}\mathcal{U}{;}&\quad&
r_k{-}\max_{j{:}j{\in}\mathcal{K}{\setminus}k}(r_j{-}a_{kj})^+{\geq}0{,}k{\in}\mathcal{U}{.}\label{eq:polyhedral}
\end{IEEEeqnarray}

\begin{myremark}\label{rmk:ZFBF_region}
Note that the DoF region achieved by ZFBF with power control can be obtained by removing the inequalities related to the common messages, i.e., \eqref{eq:DRS_KUpc}, and setting $d_k{=}d_k^p$.
\end{myremark}
\begin{figure}[t]
\renewcommand{\captionfont}{\small}
\captionstyle{center}
\centering
\subfigure[CSIT quality table]{
                \centering
                \includegraphics[width=0.2\textwidth,height=2.5cm]{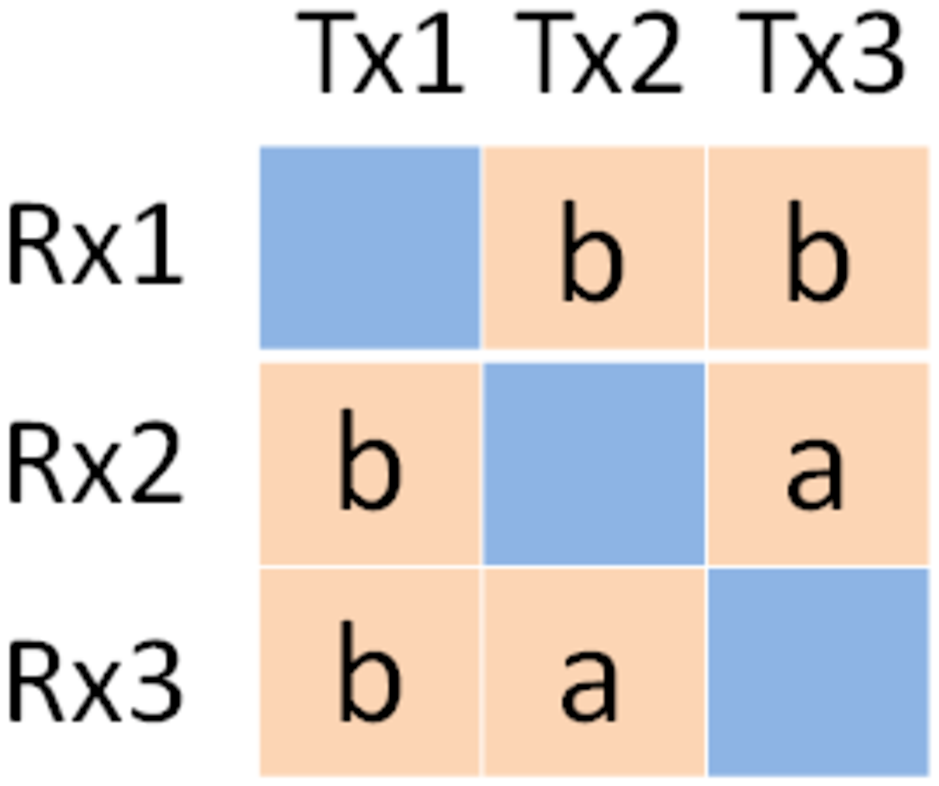}
                \label{fig:3user_het}
        }
        \subfigure[Weighted-sum interpretation]{
                \centering
                \includegraphics[width=0.4\textwidth,height=2.5cm]{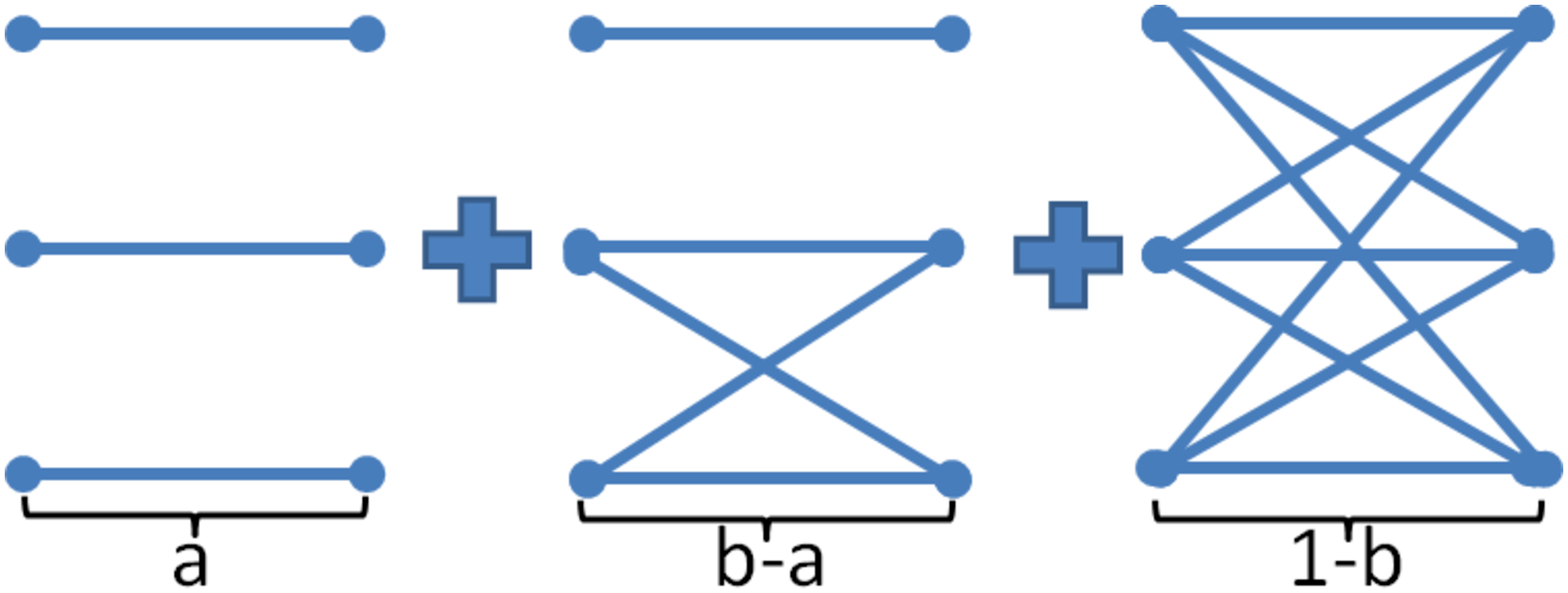}
                \label{fig:3user_het_wsi}
        }
\caption{$3$-cell IC with hierarchical CSIT quality topology, where $0{\leq}a{\leq}b{\leq}1$.}
\end{figure}

To better understand this achievable region, let us look at the example illustrated in Figure \ref{fig:3user_het}, where $0{\leq}a{\leq}b{\leq}1$. For convenience, we let $d^c{=}\sum_{k{=}1}^{3}d_k^c$. For $\mathcal{U}{=}\{1{,}2{,}3\}$, $\{2{,}3\}$, $\{1{,}3\}$, $\{1{,}2\}$, $\{3\}$, $\{2\}$ and $\{1\}$, the corresponding $\mathcal{D}_{RS}(\mathcal{U})$ are given by
\begin{IEEEeqnarray}{rcl}
\mathcal{D}_{RS}(\{1{,}2{,}3\})&{=}&\left\{0{\leq}d_k^c{\leq}1{,}0{\leq}d_k^p{\leq}1{,}0{\leq}d_k^p{+}d^c{\leq}1{,}
{\forall}k{\in}\{1{,}2{,}3\}{,}\right.\nonumber\\
&&\left.d_1^p{+}d_2^p{\leq}2b{,}d_1^p{+}d_2^p{+}d^c{\leq}1{+}b{,}d_1^p{+}d_3^p{\leq}2b{,}d_1^p{+}d_3^p{+}d^c{\leq}1{+}b{,}\right.\nonumber\\
&&\left.d_2^p{+}d_3^p{\leq}2a{,}d_2^p{+}d_3^p{+}d^c{\leq}1{+}a{,}d_1^p{+}d_2^p{+}d_3^p{\leq}2b{+}a{,} d_1^p{+}d_2^p{+}d_3^p{+}d^c{\leq}1{+}b{+}a\right\}{,}\label{eq:Drs_empty}\\
\mathcal{D}_{RS}(\{2{,}3\})&{=}&\left\{d_1^p{=}0{,}0{\leq}d_k^c{\leq}1{,}{\forall}k{\in}\{1{,}2{,}3\}{,} 0{\leq}d_k^p{\leq}1{,}0{\leq}d_k^p{+}d^c{\leq}1{,}{\forall}k{\in}\{2{,}3\}{,}\right.\nonumber\\
&&\left.d_2^p{+}d_3^p{\leq}2a{,}d_2^p{+}d_3^p{+}d^c{\leq}1{+}a\right\}{,}\label{eq:Drs_1}\\
\mathcal{D}_{RS}(\{1{,}3\})&{=}&\left\{d_2^p{=}0{,}0{\leq}d_k^c{\leq}1{,}{\forall}k{\in}\{1{,}2{,}3\}{,} 0{\leq}d_k^p{\leq}1{,}0{\leq}d_k^p{+}d^c{\leq}1{,}{\forall}k{\in}\{1{,}3\}{,}\right.\nonumber\\
&&\left.d_1^p{+}d_3^p{\leq}2b{,}d_1^p{+}d_3^p{+}d^c{\leq}1{+}b\right\}{,}\label{eq:Drs_2}
\end{IEEEeqnarray}
\begin{IEEEeqnarray}{rcl}
\mathcal{D}_{RS}(\{1{,}2\})&{=}&\left\{d_3^p{=}0{,}0{\leq}d_k^c{\leq}1{,} {\forall}k{\in}\{1{,}2{,}3\}{,} 0{\leq}d_k^p{\leq}1{,}0{\leq}d_k^p{+}d^c{\leq}1{,}{\forall}k{\in}\{1{,}2\}{,}\right.\nonumber\\
&&\left.d_1^p{+}d_2^p{\leq}2b{,}d_1^p{+}d_2^p{+}d^c{\leq}1{+}b\right\}{,}\label{eq:Drs_3}\\
\mathcal{D}_{RS}(\{3\})&{=}&\left\{d_1^p{=}d_2^p{=}0{,}0{\leq}d_k^c{\leq}1{,}{\forall}k{\in}\{1{,}2{,}3\}{,} 0{\leq}d_3^p{\leq}1{,}0{\leq}d_3^p{+}d^c{\leq}1\right\}{,}\label{eq:Drs_12}\\
\mathcal{D}_{RS}(\{2\})&{=}&\left\{d_1^p{=}d_3^p{=}0{,}0{\leq}d_k^c{\leq}1{,}{\forall}k{=}\{1{,}2{,}3\}{,} 0{\leq}d_2^p{\leq}1{,}0{\leq}d_2^p{+}d^c{\leq}1\right\}{,}\label{eq:Drs_13}\\
\mathcal{D}_{RS}(\{1\})&{=}&\left\{d_2^p{=}d_3^p{=}0{,}0{\leq}d_k^c{\leq}1{,}{\forall}k{\in}\{1{,}2{,}3\}{,} 0{\leq}d_1^p{\leq}1{,}0{\leq}d_1^p{+}d^c{\leq}1\right\}{,}\label{eq:Drs_23}
\end{IEEEeqnarray}
respectively. Using \eqref{eq:Drs_empty} through to \eqref{eq:Drs_23}, it can be verified that the maximum sum DoF $\sum_{k{=}1}^{3}d_k^c{+}d_k^p$ achieved by RS is $\max\{1{+}2a{,}1{+}b\}$. When $1{+}2a{\geq}1{+}b$, the sum DoF $1{+}2a$ is achievable by taking $r_1{=}r_2{=}r_3{=}a$; otherwise, the sum DoF $1{+}b$ is achieved with $r_1{=}r_2{=}b$ and $r_3{=}0$.

According to Remark \ref{rmk:ZFBF_region}, the DoF region achieved by ZFBF with power control can be obtained by removing the inequalities of $d^c$. Then, it can be verified that the sum DoF achieved by ZFBF with power control is $\max\{2b{,}\min\{1{+}2a{,}2b{+}a\}\}$. When $1{+}2a{\geq}2b{+}a$, the sum DoF $2b{+}a$ is achievable by choosing $(r_1{,}r_2{,}r_3){=}(2b{-}a{,}b{,}b)$; when $2b{\leq}1{+}2a{\leq}2b{+}a$, the sum DoF $1{+}2a$ is achievable by choosing $(r_1{,}r_2{,}r_3){=}(1{,}1{-}b{+}a{,}1{-}b{+}a)$; when $1{+}2a{\leq}2b$, the sum DoF $2b$ is achieved with $(r_1{,}r_2{,}r_3){=}(1{,}1{,}0)$.

By comparing the sum DoF achieved by ZFBF with power control and the sum DoF achieved by RS, we see that RS offers DoF gain except in the case $1{+}b{\leq}1{+}2a{\leq}2b{+}a$.

Next, considering ZFBF with power control and RS with common message multicasting as benchmark schemes, we move on to propose a novel transmission strategy that yields a greater DoF region in the fully connected $K$-cell MISO IC with equal SNR and INR and with an arbitrary CSIT quality topology. 

\section{Topological Rate-Splitting}\label{sec:TRS}
In this section, we firstly introduce the idea of Topological Rate-splitting focusing on the example in Figure \ref{fig:3user_het}. Secondly, we propose the generalized framework of the TRS motivated by a novel weighted-sum interpretation of the fully connected MISO IC with CSIT quality topology $\mathcal{A}$. Then, the sum DoF achieved by the TRS scheme is studied using graph theory tools.

\subsection{Toy Example}
Focusing on the example in Figure \ref{fig:3user_het}, we propose a simple TRS scheme that yields a greater sum DoF than RS and ZFBF with power control. Similar to RS, we consider that each transmitter uses power $P^a$ to unicast the private messages along ZF-precoders (this power allocation policy achieves the maximal sum DoF of RS when $1{+}2a{\geq}1{+}b$). Unlike RS, the remaining power $P{-}P^a$ is further split into two parts $P^b{-}P^a$ and $P{-}P^b$ for common message groupcasting and common message multicasting, respectively. To be specific, the transmission block for the common messages is designed as follows.

Firstly, with power $P^b{-}P^a$, we see that the interference from transmitter $1$ to user $2$ and user $3$, the interference from transmitter $2$ to user $1$ and the interference from transmitter $3$ to user $1$, can be forced within the noise power via ZFBF, because the CSIT quality of those links $a_{21}{=}a_{31}{=}a_{12}{=}a_{13}{=}b$ are sufficiently good. By doing so, the MISO IC becomes a partially connected network with two cross links $\mathbf{h}_{23}$ and $\mathbf{h}_{32}$ as illustrated in Figure \ref{fig:3user_het_wsi} (see the figure in the middle). In such a network, transmitter $1$ can deliver one message $w_1^2$ to user 1 without mixing with the messages transmitted by other transmitters. At the same time, transmitter $2$ and $3$ are able to deliver \emph{group common messages} to be decoded by user $2$ and user $3$, without mixing with $w_1^2$. Here, as we design TRS from a sum DoF perspective, for convenience, we consider that transmitter $2$ delivers a group common message $w_2^2$ while transmitter $3$ does not transmit group common message.

Secondly, with the remaining power $P{-}P^b$, as the CSIT qualities are not good enough, we see that no interference can be drowned into the noise at any user via ZFBF. This fact corresponds to a fully connected network shown in Figure \ref{fig:3user_het_wsi} (the right-most figure). Then, we consider that transmitter $1$ multicasts one common message $w_1^3$ to be decoded by all users.

Accordingly, the transmitted signals write as
\begin{IEEEeqnarray}{rcl}\label{eq:s_3user_HRS}
\mathbf{s}_1&{=}&\underbrace{\mathbf{p}_1^3w_1^3}_{P{-}P^b}{+}\underbrace{\mathbf{p}_1^2w_1^2}_{P^b{-}P^a}{+}
\underbrace{\mathbf{p}_1^1w_1^1}_{P^a}{,}\label{eq:s1_3user_hrs}\\
\mathbf{s}_2&{=}&\underbrace{\mathbf{p}_2^2w_2^2}_{P^b{-}P^a}{+}\underbrace{\mathbf{p}_2^1w_2^1}_{P^a}{,}\label{eq:s2_3user_hrs}\\
\mathbf{s}_3&{=}&\underbrace{\mathbf{p}_3^1w_3^1}_{P^a}{,}\label{eq:s3_3user_hrs}
\end{IEEEeqnarray}
where $\mathbf{p}_1^2{=}\mathbf{p}_1^1{\subseteq}\Span(\hat{\mathbf{h}}_{21}^\bot{,}\hat{\mathbf{h}}_{31}^\bot)$, $\mathbf{p}_2^2{\subseteq}\Span(\hat{\mathbf{h}}_{12}^\bot)$, $\mathbf{p}_2^1{\subseteq}\Span(\hat{\mathbf{h}}_{12}^\bot{,}\hat{\mathbf{h}}_{32}^\bot)$, and $\mathbf{p}_3^1{\subseteq}\Span(\hat{\mathbf{h}}_{13}^\bot{,}\hat{\mathbf{h}}_{23}^\bot)$. The received signals are expressed as
\begin{IEEEeqnarray}{rcl}\label{eq:y_3user_HRS}
\mathbf{y}_1&{=}&\underbrace{\mathbf{h}_{11}^H\mathbf{p}_1^3w_1^3}_{P}{+}
\underbrace{\mathbf{h}_{11}^H\mathbf{p}_1^2w_1^2}_{P^b}{+}\underbrace{\mathbf{h}_{11}^H\mathbf{p}_1^1w_1^1}_{P^a}{+}
\underbrace{\mathbf{h}_{12}^H\mathbf{p}_2^2w_2^2{+}\mathbf{h}_{12}^H\mathbf{p}_2^1w_2^1{+}\mathbf{h}_{13}^H\mathbf{p}_3^1w_3^1}_{P^0}{+}n_1{,}
\label{eq:y1_3user_hrs}\\
\mathbf{y}_2&{=}&\underbrace{\mathbf{h}_{21}^H\mathbf{p}_1^3w_1^3}_{P}{+}
\underbrace{\mathbf{h}_{21}^H\mathbf{p}_1^2w_1^2{+}\mathbf{h}_{21}^H\mathbf{p}_1^1w_1^1}_{P^0}{+}
\underbrace{\mathbf{h}_{22}^H\mathbf{p}_2^2w_2^2}_{P^b}{+}\underbrace{\mathbf{h}_{22}^H\mathbf{p}_2^1w_2^1}_{P^a}{+}
\underbrace{\mathbf{h}_{23}^H\mathbf{p}_3^1w_3^1}_{P^0}{+}n_2{,}
\label{eq:y2_3user_hrs}\\
\mathbf{y}_3&{=}&\underbrace{\mathbf{h}_{31}^H\mathbf{p}_1^3w_1^3}_{P}{+}
\underbrace{\mathbf{h}_{31}^H\mathbf{p}_1^2w_1^2{+}\mathbf{h}_{31}^H\mathbf{p}_1^1w_1^1}_{P^0}{+}
\underbrace{\mathbf{h}_{32}^H\mathbf{p}_2^2w_2^2}_{P^b}{+}\underbrace{\mathbf{h}_{32}^H\mathbf{p}_2^1w_2^1}_{P^0}{+}
\underbrace{\mathbf{h}_{33}^H\mathbf{p}_3^1w_3^1}_{P^a}{+}n_3{,}
\label{eq:y3_3user_hrs}
\end{IEEEeqnarray}
where all the undesired messages are drowned into the noise. The decoding procedure starts from the messages with the highest received power level and then downwards using SIC. The DoF achieved by the common messages are $d_1^2{=}d_2^2{=}b{-}a$ and $d_1^3{=}1{-}b$. Then, it is straightforward that the sum DoF of the common messages $1{+}b{-}2a$ is greater than that achieved in RS $1{-}a$. Counting the DoF achieved by the private messages, the sum DoF achieved by TRS is $1{+}b{+}a$, outperforming $1{+}2a$ achieved by RS.

\begin{myremark}
The beauty of the TRS approach above lies in the multi-layer structure. With ZF-precoders and properly assigned power levels, the CSIT quality topology in Figure \ref{fig:3user_het} is interpreted as a series of network topologies in Figure \ref{fig:3user_het_wsi}. As shown, the left, middle and right figures respectively represent the networks where the private message unicasting, common message groupcasting and common message multicasting are performed. The weights underneath stand for their separations in the power domain. This procedure is called weighted-sum interpretation, which helps us generating common messages to be decoded by a small number of users rather than all users. However, in RS, the common messages are always to be decoded by all the users, which essentially limits the sum DoF performance.
\end{myremark}

\subsection{Building the Generalized Transmission Block}\label{sec:trans_block}
Motivated by the toy example, we present the generalized transmission block of TRS. We describe the TRS approach focusing on the active user subset $\mathcal{S}{\subseteq}\mathcal{K}$, while the remaining users are made silent. 

In TRS, each active transmitter divides the message intended for its corresponding user into $N{=}L{+}2$ parts, i.e., $\mathcal{W}_k{\triangleq}\{w_k^1{,}w_k^2{,}\cdots{,}w_k^N\}$, ${\forall}k{\in}\mathcal{S}$. The definition of $L$ will be introduced later on. Letting $\mathbf{p}_k^i$ denote the precoder and $P_{k{,}i}$ denote the power chosen for a certain message $w_k^i$, the signal transmitted by transmitter $k$ can be expressed as
\begin{IEEEeqnarray}{rcl}
\mathbf{s}_k&{=}&\sum_{i{=}1}^{L{+}2}\underbrace{\mathbf{p}_k^iw_k^i}_{P_{k{,}i}}{,}{\forall}k{\in}\mathcal{S}{.}\label{eq:sk}
\end{IEEEeqnarray}

\underline{\emph{Private message layer}}: Message $w_k^1$ is regarded as a private message intended for user $k$ and is to be decoded by user $k$ only. It is transmitted along a ZF-precoder and is unicast with a fraction of the total power as
\begin{IEEEeqnarray}{rcl}
\mathbf{p}_k^1{\subseteq}\Span(\{\hat{\mathbf{h}}_{jk}^\bot\}_{{\forall}j{\in}\mathcal{S}{\setminus}k}),&\quad&
P_{k{,}1}{=}P^{r_k}{,}{\forall}k{\in}\mathcal{S}{.}\label{eq:Ppk1}
\end{IEEEeqnarray}

\underline{\emph{Common message layer}}: The remaining power $P{-}P^{r_k}$ at each user is employed to deliver the $L{+}1$ common messages $w_k^i{,}i{=}2{,}\cdots{,}L{+}2$. The power allocated to each common message $w_k^i$ and its precoder are obtained based on the CSIT qualities.

Firstly, as only the users in $\mathcal{S}$ are active, we obtain a subset $\mathcal{A}(\mathcal{S}){\subseteq}\mathcal{A}$ such that $a_{kj}{\in}\mathcal{A}(\mathcal{S})$ if and only if $k{,}j{\in}\mathcal{S}$. Secondly, let $\mathcal{A}(\mathbf{r}{,}\mathcal{S}){\subseteq}\mathcal{A}(\mathcal{S})$ denote the set formed by all the elements of $\mathcal{A}(\mathcal{S})$ that are greater than $r_0{\triangleq}\max_{k{\in}\mathcal{S}}r_k$, i.e., $\mathcal{A}(\mathbf{r}{,}\mathcal{S}){\triangleq}\{a_{kj}\}_{{\forall}a_{kj}{\in}\mathcal{A}(\mathcal{S}){,}a_{kj}{>}r_0}$. Thirdly, letting $L$ denote the number of different values of $\mathcal{A}(\mathbf{r}{,}\mathcal{S})$, we represent these $L$ values by $a_{\pi(1)}{,}\cdots{,}a_{\pi(L)}$, which satisfy $a_{\pi(1)}{<}a_{\pi(2)}{<}{\cdots}{<}a_{\pi(L)}$. Besides, for convenience, we define $a_{\pi(L{+}1)}{=}1$. Using these $L{+}1$ variables $a_{\pi(1)}{,}\cdots{,}a_{\pi(L{+}1)}$, we divide the remaining power $P{-}P^{r_k}$ at each transmitter into $L{+}1$ power levels, i.e., $P^{a_{\pi(1)}}{-}P^{r_k}$, $P^{a_{\pi(2)}}{-}P^{a_{\pi(1)}}$, $\cdots$, $P^{a_{\pi(L{+}1)}}{-}P^{a_{\pi(L)}}$.

Then, we assign these $L{+}1$ power levels to the common messages $w_k^i$, $2{\leq}i{\leq}L{+}2$, and choose a ZF-precoder for each of them as
\begin{IEEEeqnarray}{ccl}
P_{k{,}2}{=}P^{a_{\pi(1)}}{-}P^{r_k}{,}\quad
P_{k{,}i}{=}P^{a_{\pi(i{-}1)}}{-}P^{a_{\pi(i{-}2)}}{,}3{\leq}i{\leq}L{+}2{;}\label{eq:Pki}\\
\mathbf{p}_k^i{\subseteq} \Span\left(\{\hat{\mathbf{h}}_{jk}^\bot\}_{{\forall}j{\in}\mathcal{S}{\setminus}\mathcal{R}_k^i(\mathcal{S}{,}\mathbf{r})}\right){,}
\text{ where }
\mathcal{R}_k^i(\mathcal{S}{,}\mathbf{r}){\triangleq} \{j{:}j{\in}\mathcal{S}{\setminus}k,a_{jk}{<}a_{\pi(i{-}1)}\}{\cup}k.\label{eq:pki}
\end{IEEEeqnarray}
Such a precoder and power allocation policy suggest that $w_k^i$ is a \emph{group common message} to be decoded by the group of users $\mathcal{R}_k^i(\mathcal{S}{,}\mathbf{r})$, while it is drowned into the noise at other users ${\forall}j{\in}\mathcal{S}{\setminus}\mathcal{R}_k^i(\mathcal{S}{,}\mathbf{r})$.

With the precoders and power allocation policy given in \eqref{eq:Ppk1} through to \eqref{eq:pki}, the signal received by user $k$ writes as
\begin{IEEEeqnarray}{rcl}
y_k&{=}&\sum_{j{\in}\mathcal{S}}\sum_{i{=}1}^{L{+}2}\mathbf{h}_{kj}^H\mathbf{p}_j^iw_j^i{+}n_k\label{eq:yk}\\
&{=}&\sum_{i{=}2}^{L{+}2}\left(\underbrace{\mathbf{h}_{kk}^H\mathbf{p}_k^iw_k^i}_{P^{a_{\pi(i{-}1)}}}{+}
\sum_{\stackrel{j{:}{\forall}j{\in}\mathcal{S}{\setminus}k{,}}{a_{kj}{<}a_{\pi(i{-}1)}}} \underbrace{\mathbf{h}_{kj}^H\mathbf{p}_j^iw_j^i}_{P^{a_{\pi(i{-}1)}}}{+}
\sum_{\stackrel{j{:}{\forall}j{\in}\mathcal{S}{\setminus}k{,}}{a_{kj}{\geq}a_{\pi(i{-}1)}}} \underbrace{\mathbf{h}_{kj}^H\mathbf{p}_j^iw_j^i}_{P^{a_{\pi(i{-}1)}{-}a_{kj}}}\right){+}
\label{eq:yk2toLplus2}\\
&&\underbrace{\mathbf{h}_{kk}^H\mathbf{p}_k^1w_k^1}_{P^{r_k}}{+}\sum_{{\forall}j{\in}\mathcal{S}{\setminus}k}
\underbrace{\mathbf{h}_{kj}^H\mathbf{p}_j^1w_j^1}_{P^{r_j{-}a_{kj}}}{+}\underbrace{n_k}_{P^0}{,}\label{eq:yk1}
\end{IEEEeqnarray}
where the quantities underneath stand for the approximated received power when $P{\to}\infty$. As expressed in \eqref{eq:yk2toLplus2}, if the CSIT quality of the cross link $\mathbf{h}_{kj}$ is greater than or equal to the allocated power level, i.e., $a_{kj}{>}a_{\pi(i{-}1)}$, the common message $w_j^i{,}i{\geq}2$, is drowned into the noise at user $k$ due to ZFBF; otherwise, $w_j^i$ is received by user $k$ with power $P^{a_{\pi(i{-}1)}}$. As expressed in \eqref{eq:yk1}, the undesired private message $w_j^1{,}{\forall}j{\in}\mathcal{S}{\setminus}k$ is received by user $k$ with power $P^{r_k{-}a_{kj}}$. If $a_{kj}{\geq}r_k$, $w_j^1$ is drowned into the noise; otherwise, $w_j^1$ becomes an undesirable interference overheard by user $k$.

The decoding procedure is performed by SIC. Let us focus on the received signal in \eqref{eq:yk}. Firstly, user $k$ decodes common messages $w_k^{L{+}2}$ and $\{w_j^{L{+}2}\}_{j{:}a_{kj}{<}a_{\pi(L{+}1)}}$ by treating all the other messages as noise. Secondly, after removing those recovered messages, user $k$ decodes $w_k^{L{+}1}$ and $\{w_j^{L{+}1}\}_{j{:}a_{kj}{<}a_{\pi(L)}}$, by treating all the other messages with lower received power as noise. This procedure runs for $L{+}1$ rounds till all the common messages are recovered. At last, user $k$ decodes its desired private message $w_k^1$ by treating the undesired private messages as noise.

For convenience, let us denote the set of common messages decoded by user $k$ in a certain round of SIC by
\begin{IEEEeqnarray}{rcl}
\mathcal{T}_k^i(\mathcal{S}{,}\mathbf{r})&{\triangleq}& w_k^i\cup\{w_j^i\}_{j{:}{\forall}j{\in}\mathcal{S}{\setminus}k{,}a_{kj}{<}a_{\pi(i{-}1)}}{,}\text{ where }2{\leq}i{\leq}L{+}2{.}\label{eq:Tki}
\end{IEEEeqnarray}
Then, the DoF region achieved by the proposed TRS scheme, denoted by $\mathcal{D}_{TRS}$, is stated below.
\begin{myprop}\label{prop:DoFtuple}
In a fully connected $K$-cell IC with equal SNR and INR and with CSIT quality topology $\mathcal{A}$, the DoF region achieved by the proposed TRS scheme lies in
\begin{IEEEeqnarray}{rcl}
\mathcal{D}_{TRS}&{=}&\bigcup_{{\forall}\mathcal{S}{\in}\mathcal{K}{,}{\forall}\mathbf{r}} \mathcal{D}_{TRS}(\mathcal{S}{,}\mathbf{r}){,}
\end{IEEEeqnarray}
where $\mathcal{D}_{TRS}(\mathcal{S}{,}\mathbf{r})$ is the DoF region achieved by TRS with active user subset $\mathcal{S}$ and power allocation policy $\mathbf{r}$ for the private messages. It is the set of the DoF tuples $(d_1{,}\cdots{,}d_K){=}\sum_{i{=}1}^{L{+}2}(d_k^i{,}\cdots{,}d_K^i)$ such that
\begin{IEEEeqnarray}{lcl}
d_k^1{=}0{,}{\forall}k{\in}\mathcal{K}{\setminus}\mathcal{S}{;}\, 0{\leq}d_k^1{\leq}\left(r_k{-}\max_{j{\in}\mathcal{S}{\setminus}k}(r_j{-}a_{kj})^+\right)^+{,}{\forall}k{\in}\mathcal{S}{;}\label{eq:DTRS_1}\\
d_k^2{=}0{,}{\forall}k{\in}\mathcal{K}{\setminus}\mathcal{S}{;}\, 0{\leq}d_k^2{,} \sum_{{\forall}j{:}w_j^2{\in}\mathcal{T}_k^2(\mathcal{S}{,}\mathbf{r})}d_j^2{\leq} a_{\pi(1)}{-}\max\{r_k{,}\max_{j{\in}\mathcal{S}{\setminus}k}r_j{-}a_{kj}\}{,}{\forall}k{\in}\mathcal{S}{;}\label{eq:DTRS_2}\\
d_k^i{=}0{,}{\forall}k{\in}\mathcal{K}{\setminus}\mathcal{S}{;}\, 0{\leq}d_k^i{,} \sum_{{\forall}j{:}w_j^i{\in}\mathcal{T}_k^i(\mathcal{S}{,}\mathbf{r})}d_j^i{\leq}a_{\pi(i{-}1)}{-}a_{\pi(i{-}2)}{,} {\forall}k{\in}\mathcal{S}{,}3{\leq}i{\leq}L{+}2{,}\label{eq:DTRS_3}
\end{IEEEeqnarray}
and $\mathcal{T}_k^i(\mathcal{S}{,}\mathbf{r})$, $2{\leq}i{\leq}L{+}2$, is defined in \eqref{eq:Tki} as a function of $\mathcal{S}$ and $\mathbf{r}$.
\end{myprop}
\begin{proof}
  see Appendix C.
\end{proof}

We point out that it is cumbersome to obtain a concise expression of $\mathcal{D}_{TRS}$ by eliminating the variables $\mathbf{r}$. This is because the DoF of the common messages transmitted in each power layer are characterized by $|\mathcal{S}|$ different inequalities, which strongly depend on the CSIT quality topologies (see \eqref{eq:DTRS_2} and \eqref{eq:DTRS_3}).

In the rest of this section, we consider an inner-bound $\bar{\mathcal{D}}_{TRS}(\mathcal{S}{,}\mathbf{r}){\subseteq}\mathcal{D}_{TRS}(\mathcal{S}{,}\mathbf{r})$, obtained by replacing \eqref{eq:DTRS_2} with
\begin{IEEEeqnarray}{rcl}
d_k^2&{=}&0{,}{\forall}k{\in}\mathcal{K}{\setminus}\mathcal{S}{;}\, 0{\leq}d_k^2{,}  \sum_{{\forall}j{:}w_j^2{\in}\mathcal{T}_k^2(\mathcal{S}{,}\mathbf{r})}d_j^2{\leq} a_{\pi(1)}{-}r_0{,}{\forall}k{\in}\mathcal{S}{;}\label{eq:DTRS_2_inner}
\end{IEEEeqnarray}
where $r_0{\triangleq}\max_{k{\in}\mathcal{S}}r_k$. When there is an even power allocation for the private messages, i.e., $r_k{=}r_j{,}{\forall}k{,}j{\in}\mathcal{S}$, we have $\bar{\mathcal{D}}_{TRS}(\mathcal{S}{,}\mathbf{r}){=}\mathcal{D}_{TRS}(\mathcal{S}{,}\mathbf{r})$. Comparing this inner-bound with the DoF region achieved by RS given in \eqref{eq:DcpRS}, we can reach the conclusion that the DoF region achieved by TRS covers that achieved by RS. To see this, let us express any achievable DoF tuple $(d_1^c{,}\cdots{,}d_K^c)$ for which \eqref{eq:DcpRS} holds as $\sum_{i{=}2}^{L{+}2}(d_1^{c{,}i}{,}\cdots{,}d_K^{c{,}i})$, where the DoF tuple $(d_1^{c{,}i}{,}\cdots{,}d_K^{c{,}i})$ are subject to $\sum_{k{\in}\mathcal{S}}d_k^{c{,}i}{\leq}a_{\pi(i{-}1)}{-}a_{\pi(i{-}2)}$ and $d_k^{c{,}i}{=}0$, ${\forall}k{\in}\mathcal{K}{\setminus}\mathcal{S}$. Then, it readily shows that the DoF tuple $(d_1^{c{,}i}{,}\cdots{,}d_K^{c{,}i})$ also lies in \eqref{eq:DTRS_2_inner} and \eqref{eq:DTRS_3}, because the summation of $d_k^i$ is taken over the set ${\forall}j{:}w_j^i{\in}\mathcal{T}_k^i(\mathcal{S}{,}\mathbf{r})$, which is a subset of $\mathcal{S}$. This fact implies that the DoF region achieved by TRS covers that achieved by RS, i.e., $\mathcal{D}_{RS}(\mathcal{S}{,}\mathbf{r}){\subseteq}\bar{\mathcal{D}}_{TRS}(\mathcal{S}{,}\mathbf{r}) {\subseteq}\mathcal{D}_{TRS}(\mathcal{S}{,}\mathbf{r})$.


\subsection{Weighted-Sum Interpretation}\label{sec:wsi}
We note that the construction of the TRS scheme is motivated by a novel weighted-sum interpretation of the CSIT quality topology as a series of network topologies. Specifically, with the power and ZF-precoders chosen for the common messages in \eqref{eq:Pki} and \eqref{eq:pki}, we observe that a transmitter $k$ is only connected to the group of users ${\forall}j{\in}\mathcal{R}_k^i(\mathcal{S}{,}\mathbf{r})$. Besides, as shown by the received signal given in \eqref{eq:yk2toLplus2}, the messages $w_j^i{\in}\mathcal{S}{\setminus}\mathcal{T}_k^i(\mathcal{S}{,}\mathbf{r})$ are forced within the noise power at user $k$. This fact implies that user $k$ is only connected to transmitters ${\forall}j{,}w_j^i{\in}\mathcal{T}_k^i(\mathcal{S}{,}\mathbf{r})$. Accordingly, this topology can be expressed using a connectivity matrix $\mathbf{M}^i(\mathcal{S}{,}\mathbf{r}){\in}\{0{,}1\}^{|\mathcal{S}|{\times}|\mathcal{S}|}$, whose element in row $k$ and column $j$, i.e., $m_{kj}$, is given by
\begin{IEEEeqnarray}{rcl}
m_{kj}^i&{=}&\left\{\begin{array}{ll}
                      1 & \text{if }w_j^i{\in}\mathcal{T}_k^i(\mathcal{S}{,}\mathbf{r}){;} \\
                      0 & \text{otherwise.}
                    \end{array}\right.\label{eq:MiT}
\end{IEEEeqnarray}
Note that the value of $m_{kj}^i$ in \eqref{eq:MiT} represents whether or not $w_j^i$ is decoded by user $k$. 

The DoF tuple \eqref{eq:DTRS_2_inner} and \eqref{eq:DTRS_3} achieved by the common messages transmitted with power layer $i$ can be interpreted as $(a_{\pi(i{-}1)}{-}a_{\pi(i{-}2)}){\times}\hat{\mathcal{D}}_{TRS}^i(\mathcal{S}{,}\mathbf{r})$, where
\begin{IEEEeqnarray}{rcl}
\hat{\mathcal{D}}_{TRS}^i(\mathcal{S}{,}\mathbf{r})&{:}&\,\hat{d}_k^2{=}0{,}{\forall}k{\in}\mathcal{K}{\setminus}\mathcal{S}{;} \, 0{\leq}\hat{d}_k^2{,} {\forall}k{\in}\mathcal{S}{,} \mathbf{M}^i(\mathcal{S}{,}\mathbf{r}){\times}\hat{\mathbf{d}}^i{\leq}\mathbf{1}_{|\mathcal{S}|}{,}2{\leq}i{\leq}L{+}2{,}\label{eq:DTRS_hat_i}
\end{IEEEeqnarray}
represents the set of DoF tuples $\hat{\mathbf{d}}^i{=}(\hat{d}_1^i{,}\cdots{,}\hat{d}_K^i)$ achieved by common message groupcasting in the partially network defined by connectivity matrix $\mathbf{M}^i(\mathcal{S}{,}\mathbf{r})$. The weights $a_{\pi(i{-}1)}{-}a_{\pi(i{-}2)}$, $i{\geq}2$, stand for the fractions of channel use of the partially connected networks in the power domain (Note that we assume $a_{\pi(0)}{=}r_0$). For clarity, let $\{\hat{w}_k^i\}_{k{\in}\mathcal{S}}$ denote the common messages transmitted in the partially connected network defined by topology $\mathbf{M}^i(\mathcal{S}{,}\mathbf{r})$. The achievable DoF of $\hat{w}_k^i$ is represented by $\hat{d}_k^i$. Then, the DoF $d_k^i$ of common message $w_k^i$ transmitted in TRS is obtained by $d_k^i{=}(a_{\pi(i{-}1)}{-}a_{\pi(i{-}2)})\hat{d}_k^i$. Consequently, the DoF region $\bar{\mathcal{D}}_{TRS}^c$ contributed by all the common messages $\sum_{i{=}2}^{L{+}2}(d_1^i{,}{\cdots}d_K^i)$ can be expressed by the weighted-sum of the DoF region achieved in the $L{+}1$ partially connected networks, i.e.,
\begin{IEEEeqnarray}{rcl}
\bar{\mathcal{D}}_{TRS}^c&{=}&\sum_{i{=}2}^{L{+}2}(a_{\pi(i{-}1)}{-}a_{\pi(i{-}2)}){\times} \hat{\mathcal{D}}_{TRS}^i(\mathcal{S}{,}\mathbf{r}){.}
\label{eq:Dweighted}
\end{IEEEeqnarray}

Similarly, when $r_k{\leq}\min_{j{\in}\mathcal{S}{\setminus}k}a_{jk}$, ${\forall}k{\in}\mathcal{S}$, the private message unicasting part is interpreted as a partially connected network formed by $|\mathcal{S}|$ parallel direct links, because all the interference is drowned into the noise.

This weighted-sum interpretation bridges the DoF region achieved TRS with the achievable DoF region in partially connected networks, thus allowing us to employ methodologies applicable for partially connected networks to analyze the DoF region achieved by TRS. Motivated by this, we study the sum DoF achieved TRS in the next subsection.

\subsection{Sum DoF from Graph Theory Perspective}\label{sec:sumDoF}
In this part, we aim to find the maximal sum DoF given the DoF region $\bar{\mathcal{D}}_{TRS}(\mathcal{S}{,}\mathbf{r})$ specified by \eqref{eq:DTRS_1}, \eqref{eq:DTRS_2_inner} and \eqref{eq:DTRS_3}. To do so, it is straightforward that the maximum DoF of the private messages achieved by the TRS scheme is $d_k^1{=}\left(r_k{-}\max_{j{\in}\mathcal{S}{\setminus}k}(r_j{-}a_{kj})^+\right)^+$. Then, the work is reduced to compute the maximum sum DoF contributed by all the common messages. As a consequence of the weighted-sum interpretation in \eqref{eq:Dweighted}, this sum DoF maximization is decoupled into a series of optimization problems
\begin{IEEEeqnarray}{rcl}
\mathcal{P}_i{:}\quad \max &\quad& \hat{d}_s^i(\mathcal{S}{,}\mathbf{r}){\triangleq}\sum_{k{\in}\mathcal{S}}\hat{d}_k^i{,}{\forall}i{=}2{,}\cdots{,}L{+}2 \label{eq:barPi}\\
\text{s.t. }&\quad& (\hat{d}_k^i)_{k{\in}\mathcal{S}}{\in}\hat{\mathcal{D}}_{TRS}^i(\mathcal{S}{,}\mathbf{r})\Rightarrow 0{\leq}\hat{d}_k^i{,}k{\in}\mathcal{S}{,} \mathbf{M}^i(\mathcal{S}{,}\mathbf{r}){\times}\hat{\mathbf{d}}^i{\leq}\mathbf{1}_{|\mathcal{S}|}{.}\label{eq:Pi_cons}
\end{IEEEeqnarray}

For convenience, we drop the variables $(\mathcal{S}{,}\mathbf{r})$ in the following analysis. As explained in Section \ref{sec:wsi}, solving the problem $\mathcal{P}_i$, $i{\geq}2$, is related to maximizing the sum DoF achieved by common message groupcasting in a partially connected network. In recent years, the DoF of a partially connected network has received lots of attentions in \cite{JafarTIM,Maleki14_IA_IndexCoding,Geng13_TIMTIN,Sun13_TIMAlt,Naderializadeh15,Xinping15_TIMCoMP}. Although all of these works look at symmetric DoF as a figure of merit, graph theory methodologies have been identified as a useful means because of its powerful ability to describe whether or not a user's message is interfered with one another. Motivated by that, we solve our problems in a similar way.

We model the partially connected network with connectivity matrix $\mathbf{M}^i$ as a hypergraph $\mathcal{H}^i(\hat{\mathcal{W}}^i{,}\mathcal{T}^i)$, where $\hat{\mathcal{W}}^i{\triangleq}\{\hat{w}_k^i\}_{k{\in}\mathcal{S}}$ is the vertex set of the hypergraph and $\mathcal{T}^i{\triangleq}\{\mathcal{T}_k^i\}_{k{\in}\mathcal{S}}$ with $\mathcal{T}_k^i$ defined in \eqref{eq:Tki} is the hyperedge set of the hypergraph. Note that a member of $\mathcal{T}^i$ is actually a subset of $\hat{\mathcal{W}}^i$. If each member of $\mathcal{T}^i$ has two vertices, e.g., $\mathcal{T}_k^i{=}\{\hat{w}_k^i{,}\hat{w}_j^i\}$, then $\mathcal{T}_k^i$ actually means an edge between $\hat{w}_k^i$ and $\hat{w}_j^i$, and the hypergraph $\mathcal{H}^i(\hat{\mathcal{W}}^i{,}\mathcal{T}^i)$ is actually a graph. When an element of $\mathcal{T}^i$ has more than two elements, i.e., $|\mathcal{T}_k^i|{\geq}3$, then $\mathcal{T}_k^i$ is called an hyperedge with $|\mathcal{T}_k^i|$ vertices.

In the following, focusing on the hypergraph $\mathcal{H}^i(\hat{\mathcal{W}}^i{,}\mathcal{T}^i)$, we interpret the optimization problem $\mathcal{P}_i$ as two classical problems in graph theory, which lead to a sub-optimal solution and the optimal solution.

\subsubsection{Orthogonal Groupcasting} \label{sec:orth_groupcast}
We firstly propose a sub-optimal solution, so called \emph{orthogonal groupcasting}, by assuming that each user only decodes at most one common message. In other words, no two of the common messages $\{\hat{w}_k^i\}_{{\forall}k{\in}\mathcal{S}}$ are received by a single user. This assumption imposes a constraint $\hat{d}_k^i{\in}\{0{,}1\}$ to the optimization problem $\mathcal{P}_i$ in \eqref{eq:barPi}.

Then, a DoF tuple $(\hat{d}_k^i)_{k{\in}\mathcal{S}}$ achieved by orthogonal groupcasting defines a subset $\mathcal{X}^i{\subseteq}\mathcal{W}^i$ which contains all the messages with DoF $1$, i.e., $\mathcal{X}^i{=}\{\hat{w}_k^i\}_{{\forall}k{,}\hat{d}_k^i{=}1}$. The sum DoF is identical to the cardinality of $\mathcal{X}^i$, i.e., $|\mathcal{X}^i|$. According to the definition of orthogonal groupcasting, this subset has the property that no two elements of $\mathcal{X}^i$ are together in the same member of $\mathcal{T}^i$. Therefore, this subset $\mathcal{X}^i{\subseteq}\hat{\mathcal{W}}^i$ is called a \emph{packing} in the hypergraph $\mathcal{H}^i(\hat{\mathcal{W}}^i{,}\mathcal{T}^i)$ \cite{FracGraTheo}. Finding the maximum sum DoF is equivalent to finding the largest size of a packing, and the largest size is defined to be the \emph{packing number} $p(\mathcal{H}^i(\hat{\mathcal{W}}^i{,}\mathcal{T}^i))$ of $\mathcal{H}^i(\hat{\mathcal{W}}^i{,}\mathcal{T}^i)$. Hence, the sum DoF achieved by the orthogonal groupcasting is stated as follows.
\begin{myprop}\label{prop:orth_groupcast}
In a fully connected $K$-cell MISO IC with equal SNR and INR and with CSIT quality topology $\mathcal{A}$, the sum DoF achieved by TRS designed with orthogonal common message groupcasting is
\begin{IEEEeqnarray}{rcl}
d_{s{,}TRS}^{\text{orth}}&{=}& \max_{{\forall}\mathcal{S}{\subseteq}\mathcal{K}{,}{\forall}\mathbf{r}}d_{s{,}TRS}^{\text{orth}}(\mathcal{S}{,}\mathbf{r}){,}
\text{ where}\\
d_{s{,}TRS}^{\text{orth}}(\mathcal{S}{,}\mathbf{r})&{=}& \sum_{k{=}1}^K\left(r_k{-}\max_{j{\in}\mathcal{S}^\prime{\setminus}k}(r_j{-}a_{kj})^+\right)^+{+} \sum_{i{=}2}^{L{+}2} (a_{\pi(i{-}1)}{-}a_{\pi(i{-}2)}){\times}p(\mathcal{H}^i(\hat{\mathcal{W}}^i{,}\mathcal{T}^i(\mathcal{S}{,}\mathbf{r}))){,}
\end{IEEEeqnarray}
where $p(\mathcal{H}^i(\hat{\mathcal{W}}^i{,}\mathcal{T}^i(\mathcal{S}{,}\mathbf{r})))$ refers to the packing number of a hypergraph $\mathcal{H}^i(\hat{\mathcal{W}}^i{,}\mathcal{T}^i(\mathcal{S}{,}\mathbf{r}))$ defined by vertex set $\hat{\mathcal{W}}^i$ and hyperedge set $\mathcal{T}^i(\mathcal{S}{,}\mathbf{r})$ defined in \eqref{eq:Tki}.
\end{myprop}

\subsubsection{Maximal groupcasting}
To find the optimal solution to problem $\mathcal{P}_i$, let us firstly look at the following problem.
\begin{IEEEeqnarray}{rcl}
\tilde{\mathcal{P}}_i{:}\quad\max&\quad& \sum_{k{\in}\mathcal{S}}\tilde{d}_k^i \label{eq:tildePi}\\
\text{s.t. }&\quad& \mathbf{M}^i{\times}\tilde{\mathbf{d}}^i{\leq}t{\times}\mathbf{1}_{|\mathcal{S}|}{,} \tilde{d}_k^i{\in}\mathbb{Z}^+{,}{\forall}k{\in}\mathcal{S}{,}\label{eq:Mdt}
\end{IEEEeqnarray}
where $t$ is a positive integer. A feasible $(\tilde{d}_k^i)_{k{\in}\mathcal{S}}$ satisfying \eqref{eq:Mdt} defines a \emph{multiset} $\mathcal{X}^i$ which contains $\hat{w}_k^i$ if $\tilde{d}_k^i{>}0$. The multiplicity\footnote{The multiset $\mathcal{X}^i$ may have multiple identical elements. For instance, one has $\mathcal{X}^i{=}\{\hat{w}_1^i{,}\hat{w}_2^i{,}\hat{w}_2^i\}$, and the multiplicity of $\hat{w}_1^i$ is $1$ and the multiplicity of $\hat{w}_2^i$ is $2$.} of $\hat{w}_k^i$ in $\mathcal{X}^i$ is $\tilde{d}_k^i$, and the sum DoF is equal to $|\mathcal{X}^i|$.

In this way, the inequality \eqref{eq:Mdt} can be interpreted as follows. For the vertices in the same member of $\mathcal{T}^i$, the sum of their multiplicity in $\mathcal{X}^i$ is smaller than or equal to $t$. According to \cite{FracGraTheo}, a multiset $\mathcal{X}^i$ with such a property is called a \emph{$t$-fold packing} of hypergraph $\mathcal{H}^i(\hat{\mathcal{W}}^i{,}\mathcal{T}^i)$. When $t{=}1$, the $t$-fold packing collapses to the packing of the hypergraph $\mathcal{H}^i(\hat{\mathcal{W}}^i{,}\mathcal{T}^i)$ that is introduced in Section \ref{sec:orth_groupcast}. Consequently, the optimization problem $\tilde{\mathcal{P}}_i$ is interpreted as finding the largest size of a $t$-fold packing, and the largest size is defined as the $t$-fold packing number $p_t(\mathcal{H}^i(\hat{\mathcal{W}}^i{,}\mathcal{T}^i))$.

So far, we are one-step closer to our objective. According to \cite[Section 1.2]{FracGraTheo}, the optimal result of Problem $\mathcal{P}_i$ in \eqref{eq:barPi} can be found using the result of Problem $\tilde{\mathcal{P}}_i$ in \eqref{eq:tildePi} by taking $t{\to}\infty$ as
\begin{IEEEeqnarray}{rcl}
p_f(\mathcal{H}^i(\hat{\mathcal{W}}^i{,}\mathcal{T}^i))&{=}& \lim_{t{\to}\infty}\frac{p_t(\mathcal{H}^i(\hat{\mathcal{W}}^i{,}\mathcal{T}^i))}{t}{.}
\label{eq:pf}
\end{IEEEeqnarray}
This quantity is called \emph{fractional packing number} of hypergraph $\mathcal{H}^i(\hat{\mathcal{W}}^i{,}\mathcal{T}^i)$. Besides, the DoF of message $\hat{w}_k^i$ is expressed as $\hat{d}_k^i{=}\lim_{t{\to}\infty}\frac{\tilde{d}_k^{i*}}{t}$, where $\tilde{d}_k^{i*}$ is the result of the $t$-fold packing problem $\tilde{\mathcal{P}}_i$.

Therefore, we may state an achievable sum DoF resulted by the maximal groupcasting as follows.
\begin{myprop}\label{prop:MCMM}
In a fully connected $K$-cell MISO IC with equal SNR and INR and with CSIT quality topology $\mathcal{A}$, the sum DoF achieved by TRS with maximal groupcasting is
\begin{IEEEeqnarray}{rcl}
d_{s{,}TRS}^{\max}&{=}& \max_{{\forall}\mathcal{S}{\subseteq}\mathcal{K}{,}{\forall}\mathbf{r}}d_{s{,}TRS}^{\max}(\mathcal{S}{,}\mathbf{r}){,}\text{ where}\\
d_{s{,}TRS}^{\max}(\mathcal{S}{,}\mathbf{r})&{=}& \sum_{k{=}1}^K\left(r_k{-}\max_{j{\in}\mathcal{S}^\prime{\setminus}k}(r_j{-}a_{kj})^+\right)^+{+} \sum_{i{=}2}^{L{+}2} (a_{\pi(i{-}1)}{-}a_{\pi(i{-}2)}){\times}p_f(\mathcal{H}^i(\hat{\mathcal{W}}^i{,}\mathcal{T}^i(\mathcal{S}{,}\mathbf{r}))){,}
\end{IEEEeqnarray}
where $p_f(\mathcal{H}^i(\hat{\mathcal{W}}^i{,}\mathcal{T}^i(\mathcal{S}{,}\mathbf{r})))$ refers to the fractional packing number of a hypergraph $\mathcal{H}^i(\hat{\mathcal{W}}^i{,}\mathcal{T}^i(\mathcal{S}{,}\mathbf{r}))$ defined by vertex set $\hat{\mathcal{W}}^i$ and hyperedge set $\mathcal{T}^i(\mathcal{S}{,}\mathbf{r})$ defined in \eqref{eq:Tki}.
\end{myprop}
\begin{figure}
  \centering
  \includegraphics[width=0.5\textwidth,height=3cm]{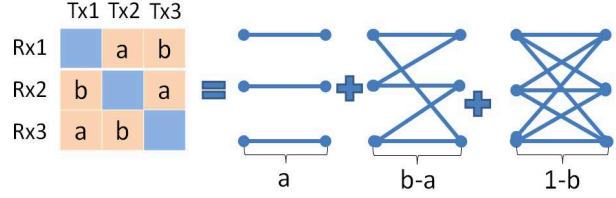}
  \caption{$3$-cell IC with a cyclic CSIT quality topology}\label{fig:3user_cyc}
\end{figure}
Note that both common message groupcasting methods suffice to achieve the sum DoF in the example illustrated in Figure \ref{fig:3user_het}. To highlight the gain offered by the maximal groupcasting, let us focus on the $3$-cell scenario with a cyclic CSIT quality topology illustrated in Figure \ref{fig:3user_cyc}.

Following the footsteps presented in Section \ref{sec:trans_block}, the transmitted signal consists of three power levels, $P^a$, $P^b{-}P^a$ and $P{-}P^b$, which are used for private message unicasting, common message groupcasting and common message multicasting. To highlight the benefit of performing maximal groupcasting, we only discuss the sum DoF achieved by the messages transmitted in the second power level.

With power $P^b{-}P^a$ and ZF-precoders, three interference links can be ``removed'', and the remaining links form a cyclic partially connected network as illustrated in Figure \ref{fig:3user_cyc}. In this network, with the orthogonal groupcasting method, only one message can be successfully transmitted, e.g., $(d_1^2{,}d_2^2{,}d_3^2){=}(b{-}a{,}0{,}0)$, $(0{,}b{-}a{,}0)$ or $(0{,}0{,}b{-}a)$. Otherwise, there will be some users receiving a mixture of two common messages, which contradicts the philosophy of the orthogonal groupcasting method. However, the maximal groupcasting method requires each user to decode multiple common messages. By doing so, although the DoF of each common message decreases, the sum DoF can be enhanced since more common messages can be transmitted. Specifically, since each user receives the mixture of two common messages, it is straightforward that the per common message DoF $\frac{b{-}a}{2}$ is achievable, thus leading to the sum DoF of $\frac{3}{2}(b{-}a)$, which outperforms $b{-}a$ achieved by orthogonal groupcasting.

Counting the DoF $3a$ achieved by the private messages and the DoF $1{-}b$ achieved by common message multicasting with power $P{-}P^b$, the sum DoF achieved by TRS designed with maximal groupcasting is $1{+}\frac{b{+}3a}{2}$. Note that this result outperforms the sum DoF $1{+}2a$ achieved by RS, and the sum DoF $\max\{a{+}b{,}3a\}$ achieved by ZFBF with power control.

Last but not the least, we point out that the sum DoF stated in Proposition \ref{prop:MCMM} yields the best result so far, because it has been shown that the DoF region stated in Proposition \ref{prop:DoFtuple} covers the DoF region achieved by RS and ZFBF with power control. Unfortunately, due to the complicated expression of the sum DoF achieved by TRS, the general sufficient and necessary condition where TRS strictly outperforms RS and ZFBF with power control is yet to be characterized. In an extreme case where the CSIT of the interference links associated to a single user have equal qualities, i.e., $a_{kj}{=}\alpha_k$, ${\forall}k{\in}\mathcal{K}$, ${\forall}j{\in}\mathcal{K}{\setminus}k$, following the footsteps presented in Section \ref{sec:trans_block}, we can see that there always exists a user who has to decode all the common messages. As a result, the sum DoF achieved by TRS is essentially impacted and is no greater than the sum DoF achieved by RS. Therefore, we claim that when the CSIT qualities of the interference links associated to each user have a larger variance, TRS is more likely to strictly outperform RS and ZFBF with power control. 

\section{Realistic Scenarios}\label{sec:real}
\begin{figure}[t]
\renewcommand{\captionfont}{\small}
\captionstyle{center}
\centering
\subfigure[CSIT quality table and network topology]{
                \centering
                \includegraphics[width=0.35\textwidth,height=3cm]{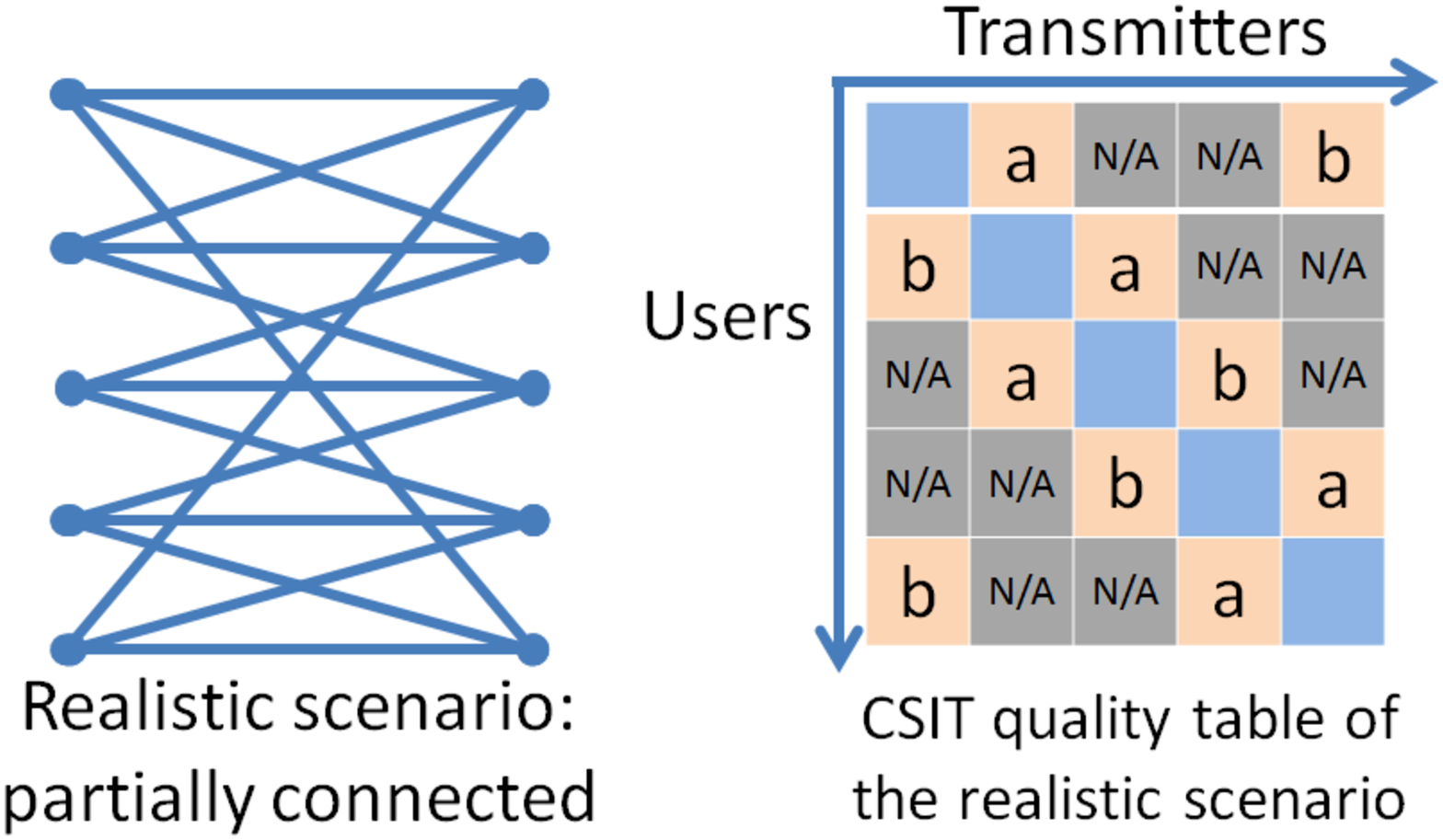}
                \label{fig:5user1}
        }
        \subfigure[Weighted-sum interpretation]{
                \centering
                \includegraphics[width=0.35\textwidth,height=2.5cm]{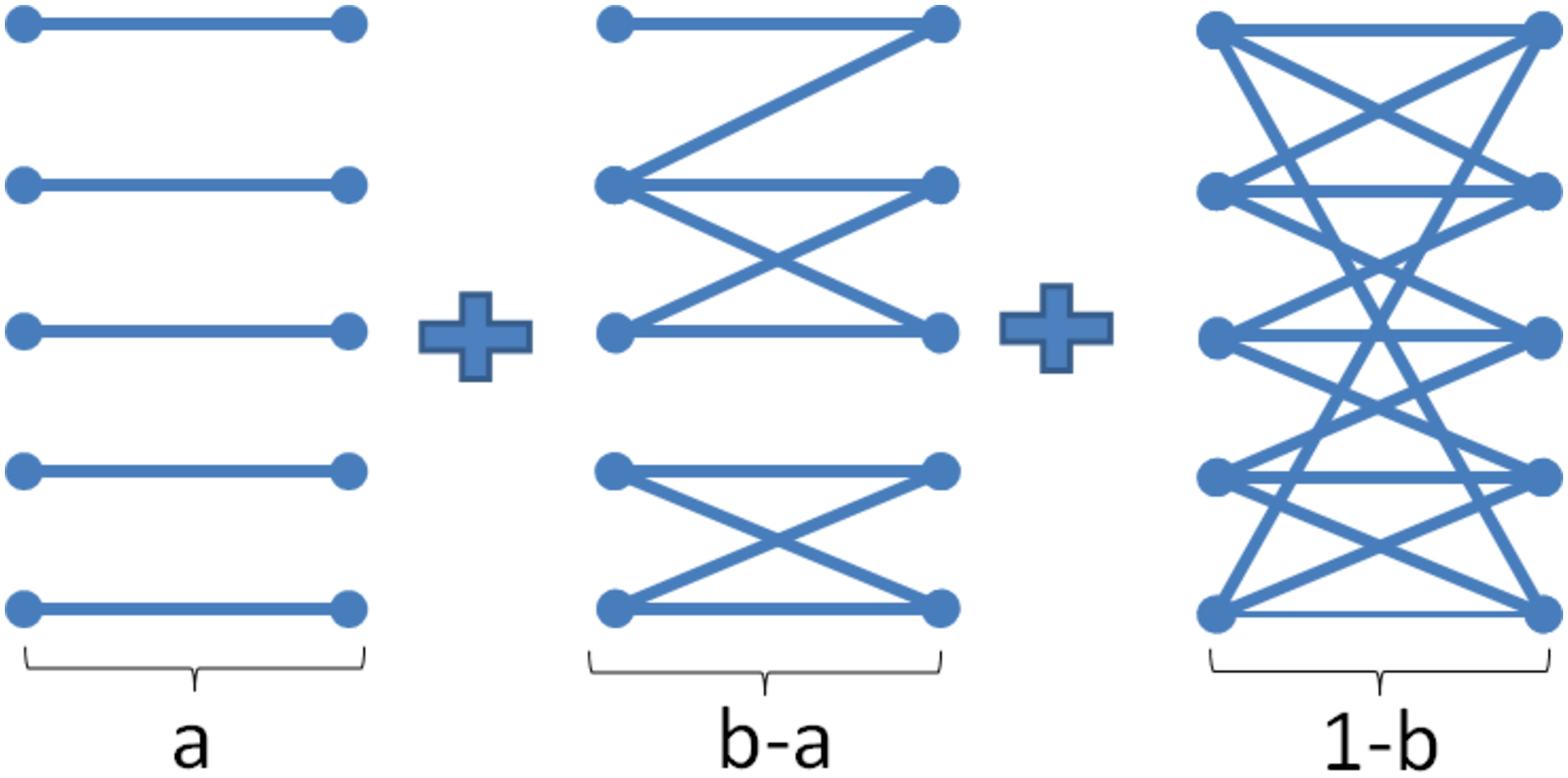}
                \label{fig:5user_wsi}
        }
\caption{$5$-user examples with realistic setting, where $0{\leq}a{\leq}b{\leq}1$.}
\end{figure}
So far, we have identified the achievability of the TRS scheme in the fully connected IC where $g_{kj}{=}1$, ${\forall}k{,}j{\in}\mathcal{K}$. In this section, we show that the philosophy of the TRS scheme is also applicable to partially connected networks with imperfect CSIT. To see this, we switch our attention to a realistic scenario in the homogeneous cellular network \cite{BrunoBook}, where each user typically only receives the signal sent by its serving transmitter, and the signals sent by two adjacent transmitters, i.e., user $k$ only receives $\mathbf{s}_k$, $\mathbf{s}_{k{+}1}$ and $\mathbf{s}_{k{-}1}$. The signals sent out by farther transmitters are assumed to be negligible due to the long distance. Note that it is assumed that user $1$ is connected to transmitter $K$, $1$ and $2$, while user $K$ is connected to transmitter $K{-}1$, $K$ and user $1$.

In the following, we firstly design a TRS approach for a class of CSIT quality topologies, where for each user, one incoming interference link has CSIT quality $b$, while the other interference link has CSIT quality $a$, i.e., either $(a_{k{,}k{+}1}{,}a_{k{,}k{-}1}){=}(a{,}b)$ or $(a_{k{,}k{+}1}{,}a_{k{,}k{-}1}){=}(b{,}a)$, ${\forall}k{\in}\mathcal{K}$. It is assumed that $a{\leq}b$. A $5$-cell example is illustrated in Figure \ref{fig:5user1}. Secondly, we find the closed-form expression of the maximal sum DoF achieved by the proposed TRS. Lastly, we compare the results with the sum DoF achieved by ZFBF with power control.

\subsection{TRS scheme}
Without sum DoF maximization, we design the TRS by the considering that all users are active and each transmitter uses power $P^a$, i.e., $r_k{=}a{,}{\forall}k{\in}\mathcal{K}$, to unicast the private message. Following the footsteps presented in Section \ref{sec:trans_block}, the transmitted signal is expressed as
\begin{IEEEeqnarray}{rcl}
\mathbf{s}_k&{=}&\underbrace{\mathbf{p}_k^3w_k^3}_{P{-}P^b}{+}\underbrace{\mathbf{p}_k^2w_k^2}_{P^b{-}P^a}{+}
\underbrace{\mathbf{p}_k^1w_k^1}_{P^a}{,}\label{eq:sk_real}
\end{IEEEeqnarray}
where $\mathbf{p}_k^3$ is a random precoder, $\mathbf{p}_k^2{\subseteq}\Span(\{\hat{\mathbf{h}}_{jk}^\bot\}_{j{=}k{+}1{,}k{-}1{,}a_{jk}{=}b})$, and $\mathbf{p}_k^1{\subseteq}\Span(\hat{\mathbf{h}}_{k{+}1{,}k}^\bot{,}\hat{\mathbf{h}}_{k{-}1{,}k}^\bot)$ are ZF-precoders. The message $w_k^1$ is a private message intended for user $k$, $w_k^2$ is a common message to be decoded by user $k$ and user $j$ for some $j{=}k{+}1{,}k{-}1{,}a_{jk}{=}a$, while $w_k^3$ is a common message to be decoded by user $k$, $k{-}1$ and $k{+}1$.

The signal received by user $k$ writes as
\begin{IEEEeqnarray}{rcl}
y_k&{=}&\sum_{j{=}k{-}1}^{k{+}1}\underbrace{\mathbf{h}_{kj}^H\mathbf{p}_j^3w_j^3}_{P}{+}
\label{eq:yk_level3}\\
&&\underbrace{\mathbf{h}_{kk}^h\mathbf{p}_k^2w_k^2}_{P^b}{+}
\sum_{j{=}k{-}1{,}k{+}1{,}a_{kj}{=}a}\underbrace{\mathbf{h}_{kj}^H\mathbf{p}_j^2w_j^2}_{P^b}{+}
\sum_{j{=}k{-}1{,}k{+}1{,}a_{kj}{=}b}\underbrace{\mathbf{h}_{kj}^H\mathbf{p}_j^2w_j^2}_{P^0}{+}
\label{eq:yk_level2}\\
&&\underbrace{\mathbf{h}_{kk}^H\mathbf{p}_k^1w_k^1}_{P^a}{+}\sum_{j{=}k{-}1{,}k{+}1}\underbrace{\mathbf{h}_{kj}^H\mathbf{p}_j^1w_j^1}_{P^0}{+}
\underbrace{n_k}_{P^0}{.}\label{eq:yk_level1}
\end{IEEEeqnarray}
The sets of the common messages that are decoded by user $k$ are defined as $\mathcal{T}_k^1{\triangleq}\{w_k^1\}$, $\mathcal{T}_k^2{\triangleq}w_k^2{\cup}w_j^2$, $j{=}k{+}1{,}k{-}1{,}a_{jk}{=}a$ and $\mathcal{T}_k^3{\triangleq}\{w_k^3{,}w_{k{+}1}^3{,}w_{k{-}1}^3\}$. By performing SIC, the achievable DoF lies in
\begin{IEEEeqnarray}{rcl}
d_k^1{\leq}a{,}\quad & \sum_{j{\in}\mathcal{T}_k^2}d_j^2{\leq}b{-}a{,}\quad&
\sum_{j{\in}\mathcal{T}_k^3}d_j^3{\leq}1{-}b{,}{\forall}k{\in}\mathcal{K{.}}\label{eq:Dreal}
\end{IEEEeqnarray}

With the definition of $\mathcal{T}_k^i$, the weighted-sum interpretation of the CSIT quality topology in Figure \ref{fig:5user1} is illustrated in Figure \ref{fig:5user_wsi}. The left, middle and right figures respectively stand for the partially connected networks where the private message unicasting, common message groupcasting and common message multicasting are performed. Next, given the achievable DoF region in \eqref{eq:Dreal}, we study the maximal achievable sum DoF.

\subsection{Sum DoF achieved by the proposed TRS}
Firstly, it is clear that the maximum sum DoF achieved by the private messages $\{w_k^1\}_{k{\in}\mathcal{K}}$ is $Ka$.

Secondly, the maximum sum DoF achieved by the common messages $\{w_k^3\}_{k{\in}\mathcal{K}}$ can be found as follows. The inequalities in \eqref{eq:Dreal} related to $w_k^3$ can be explicitly written as $d_1^3{+}d_2^3{+}d_3^3{\leq}1{-}b$, $d_2^3{+}d_3^3{+}d_4^3{\leq}1{-}b$, $\cdots$, $d_{K{-}1}^3{+}d_K^3{+}d_1^3{\leq}1{-}b$ and $d_{K}^3{+}d_1^3{+}d_2^3{\leq}1{-}b$. Summing these $K$ inequalities yields $3\sum_{k{=}1}^{K}d_k^3{\leq}{K}(1{-}b)$, leading to the sum DoF $\sum_{k{=}1}^{K}d_k^3{\leq}\frac{K}{3}(1{-}b)$. The equality holds by simply taking $d_1^3{=}d_2^3{=}\cdots{=}d_K^3{=}\frac{1{-}b}{3}$.

Thirdly, it remains to compute the maximal achievable DoF of the common messages $\{w_k^2\}_{k{\in}\mathcal{K}}$. To this end, according to the definition of set $\mathcal{T}_k^2$, we obtain a partially connected network with the topology matrix $\mathbf{M}^2$, whose elements are determined following \eqref{eq:MiT}. Specifically, if $j{=}k$ or $a_{kj}{=}a$, we have $m_{kj}{=}1$; otherwise, we have $m_{kj}{=}0$. An example of the obtained partially connected network is illustrated in Figure \ref{fig:5user_wsi}. Then, finding the sum DoF achieved by $\{w_k^2\}_{k{\in}\mathcal{K}}$ subject to \eqref{eq:Dreal} is equivalent to computing the sum DoF $\sum_{k{\in}\mathcal{K}}\hat{d}_k^2$ given $\mathbf{M}^2\hat{\mathbf{d}}^2{\leq}\mathbf{1}_K$, where $\hat{d}_k$ stands for DoF of common message $\hat{w}_k^i$ transmitted in the partially connected network defined by $\mathbf{M}^2$. The DoF $d_k^2$ achieved in TRS is obtained by $(b{-}a)\hat{d}_k^2$.

According the CSIT quality topology mentioned at the beginning of this section, we see that each set $\mathcal{T}_k^2{,}{\forall}k{\in}\mathcal{K}$ has two elements, and each row of $\mathbf{M}^2$ has two ``$1$''s. Then, following the definition introduced in Section \ref{sec:sumDoF}, the hypergraph $\mathcal{H}^2(\hat{\mathcal{W}}^2{,}\mathcal{T}^2)$ is actually a graph. A member of $\mathcal{T}^2$, i.e., $\mathcal{T}_k^i$, refers to an edge between vertex $\hat{w}_k^2$ and its neighbor $\hat{w}_j^2$ if $a_{kj}{=}a$. Next, we characterize the sum DoF $\sum_{k{\in}\mathcal{K}}\hat{d}_k^2$ by evaluating the row rank of $\mathbf{M}^2$.

\begin{figure}[t]
\renewcommand{\captionfont}{\small}
\captionstyle{center}
\centering
\subfigure[$\Rowrk(\mathbf{M}^2){=}5$]{
                \centering
                \includegraphics[width=0.3\textwidth,height=3cm]{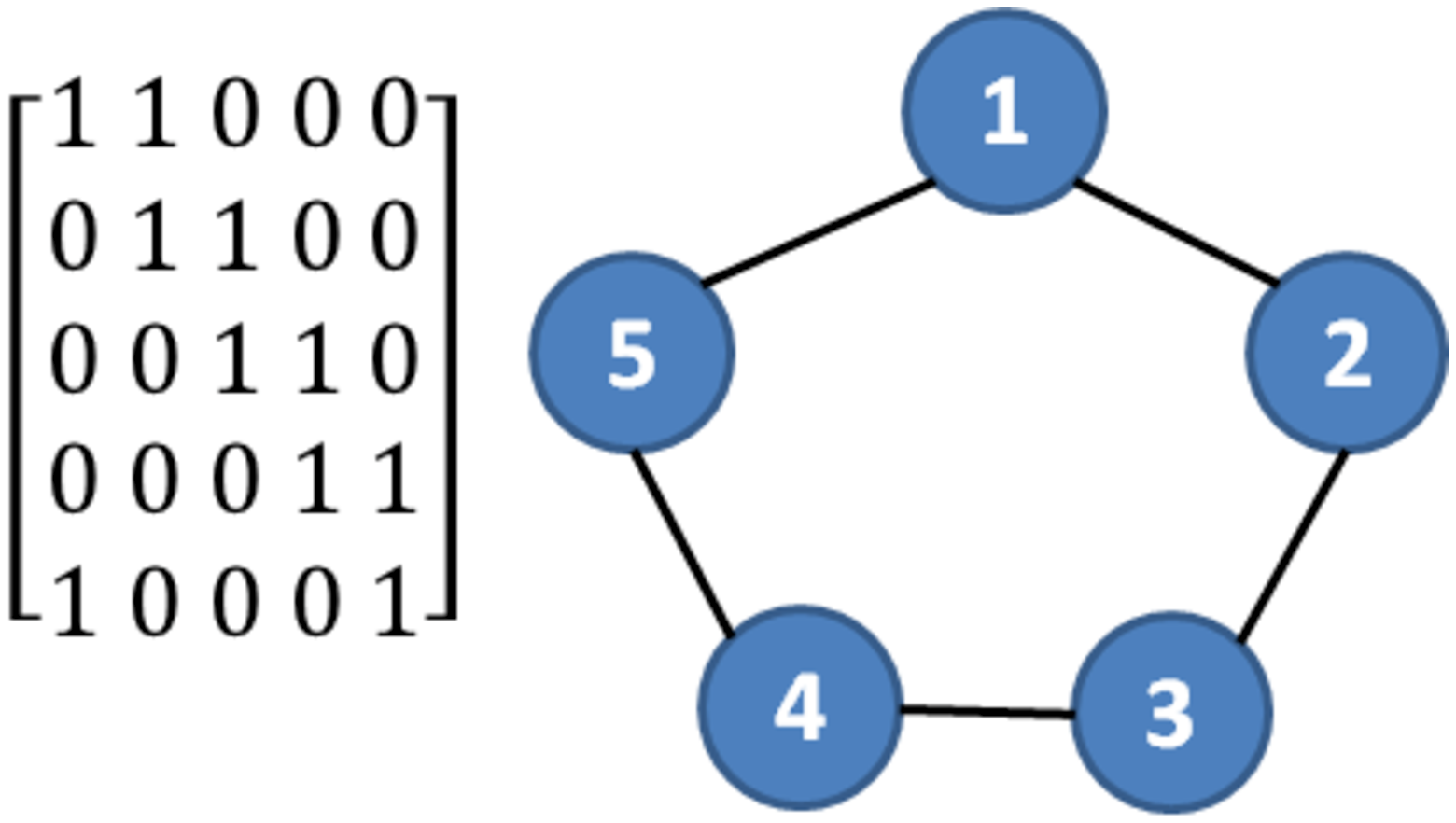}
                \label{fig:5user1_rk5}
        }
        \subfigure[$\Rowrk(\mathbf{M}^2){=}4$]{
                \centering
                \includegraphics[width=0.3\textwidth,height=3cm]{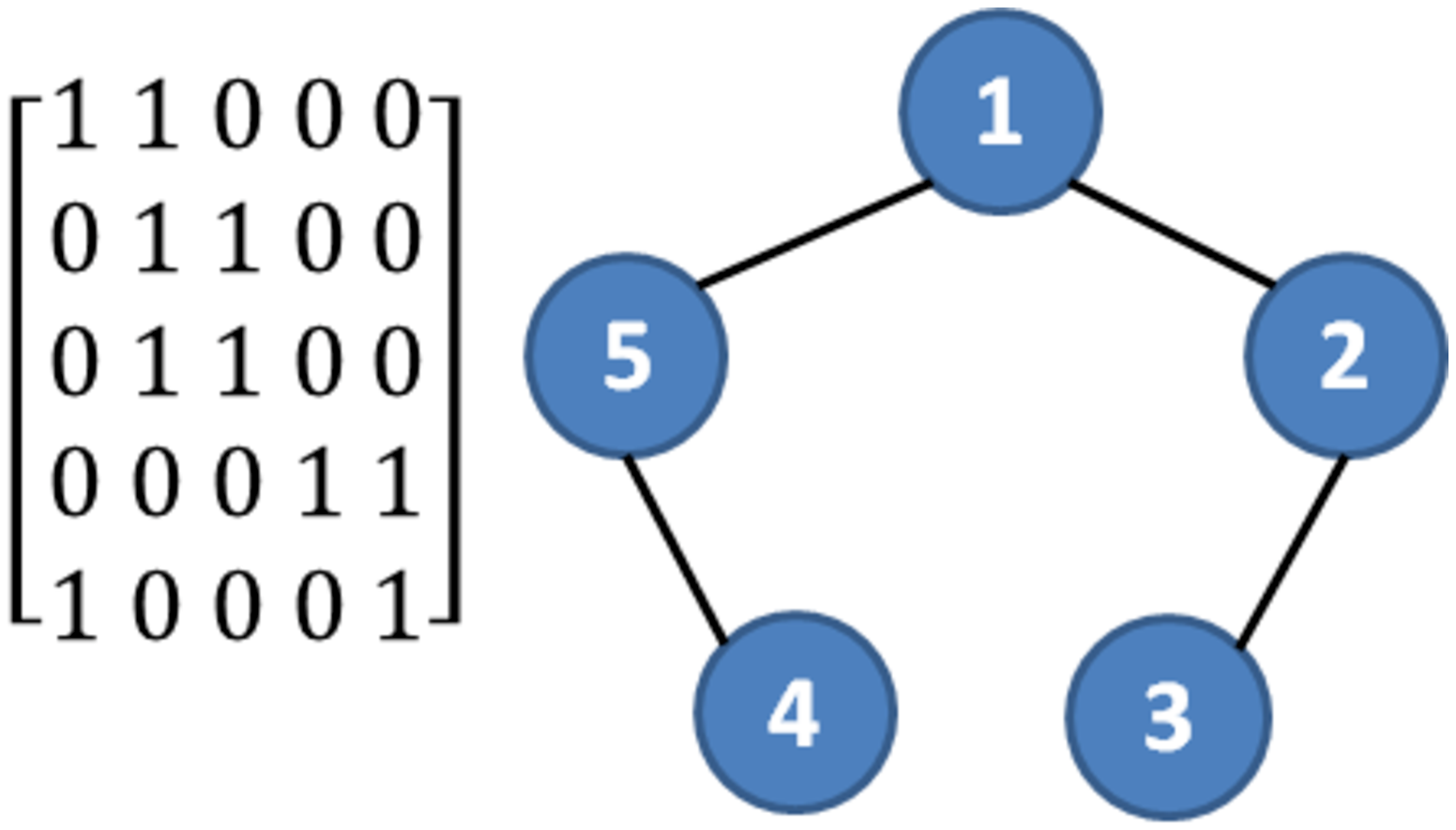}
                \label{fig:5user1_rk4}
        }
        \subfigure[$\Rowrk(\mathbf{M}^2){=}3$]{
                \centering
                \includegraphics[width=0.3\textwidth,height=3cm]{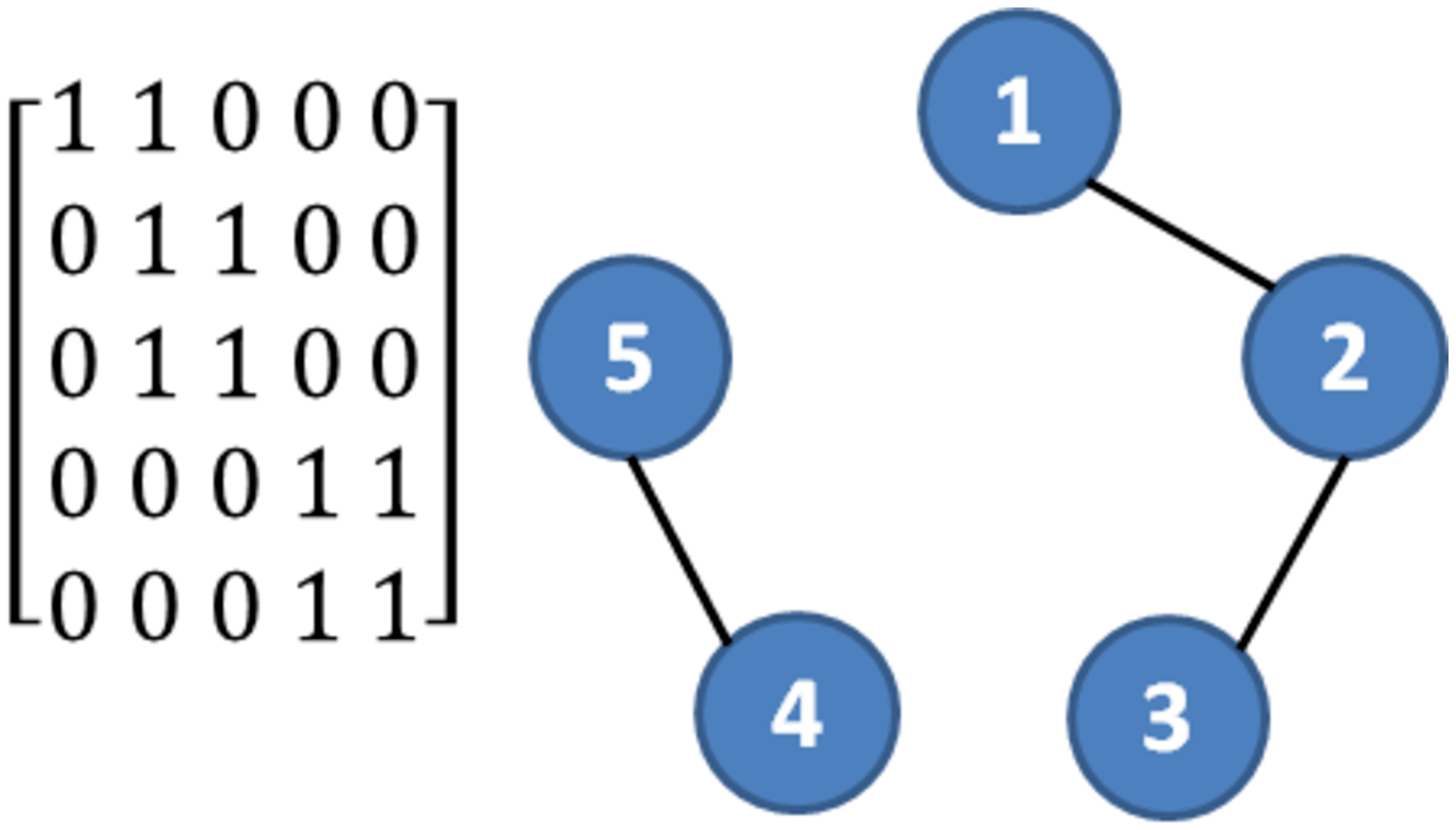}
                \label{fig:5user1_rk3}
        }
\caption{The hypergraph $\mathcal{H}^2(\mathcal{W}^2{,}\mathcal{T}^2)$.}
\end{figure}
When $\mathbf{M}^2$ has a full row rank, it means that there is no redundant inequalities in \eqref{eq:Dreal}. In other words, there is no overlapping edges in $\mathcal{T}^2$. Moreover, since a vertex $\hat{w}_k^2$ can only have an edge with either $\hat{w}_{k{-}1}^i$ or $\hat{w}_{k{+}1}^i$, the graph $\mathcal{H}^2(\hat{\mathcal{W}}^2{,}\mathcal{T}^2)$ is actually a circuit. An example is illustrated Figure \ref{fig:5user1_rk5}. It can be verified that the sum DoF of the common messages $\{\hat{w}_k^2\}_{k{\in}\mathcal{K}}$ is $\frac{K}{2}$ (obtained by adding up all the $K$ inequalities involved in $\mathbf{M}^2\hat{\mathbf{d}}^2{\leq}\mathbf{1}_K$ and dividing the sum by $2$).

When $\mathbf{M}^2$ has a deficient row rank, it means that some edges of the graph $\mathcal{H}^2(\hat{\mathcal{W}}^2{,}\mathcal{T}^2)$ are redundant. This fact breaks the circuit when $\mathbf{M}^2$ has full row rank into pieces. Clearly, if the row rank of $\mathbf{M}^2$ is $\Rowrk(\mathcal{M}^2){=}K{-}1$, the graph is a chain (see Figure \ref{fig:5user1_rk4}); if the row rank of $\mathbf{M}^2$ is $\Rowrk(\mathcal{M}^2){=}K{-}2$, the graph consists of two separated chains (see Figure \ref{fig:5user1_rk3}), then the maximum sum DoF can be computed by adding up the sum DoF achieved in each chain. Hence, when $\Rowrk(\mathcal{M}^2){=}r$, the graph has $K{-}r$ separated chains. The remaining work is to characterize the maximum sum DoF for a single chain.

Intuitively, as two connected vertices correspond to a sum DoF constraint $\hat{d}_k^2{+}\hat{d}_{k{+}1}^2{\leq}1$, the maximum sum DoF for a single chain is equal to the number of disjoint vertices. Hence, denoting the length of a chain by $K_n$, the sum DoF is $\frac{K_n}{2}$ if $K_n$ is an even number and $\frac{K_n{+}1}{2}$ if $K_n$ is an odd number. The rigorous proof is presented in Appendix D.

In general, when $\Rowrk(\mathcal{M}^2){=}r$, the sum DoF of common messages $\{\hat{w}_k^2\}_{k{\in}\mathcal{K}}$ writes as
\begin{IEEEeqnarray}{rcl}
\sum_{k{\in}\mathcal{K}}^{\hat{d}_k^2}&{=}&\left\{\begin{array}{ll}\frac{K}{2} & \text{if }r{=}K\\
\sum_{n{=}1}^{K{-}r}\frac{K_n}{2}1_{K_n\text{ is even}}{+}\frac{K_n{+}1}{2}1_{K_n\text{ is odd}}& \text{if }r{<}K\end{array}\right.
{=}\frac{K}{2}{+}\frac{\epsilon}{2}{,}\label{eq:ds2_scn1}
\end{IEEEeqnarray}
where $\epsilon$ stands for the number of chains that have odd number of vertices. Then, the maximum sum DoF achieved by $\{w_k^2\}_{k{\in}\mathcal{K}}$ transmitted in TRS is $(b{-}a)\left(\frac{K}{2}{+}\frac{\epsilon}{2}\right)$.

According to the above analysis and counting the sum DoF achieved by $\{w_k^i\}_{k{\in}\mathcal{K}}{,}i{=}1{,}3$, we state the maximum achievable sum DoF in the considered scenario as follows.
\begin{myprop}\label{prop:scn1}
In a $K$-cell MISO IC where 1) each user is connected to its closest three transmitters, and 2) the two incoming interference links associated to each user has CSIT quality $a$ and $b$ with $0{\leq}a{\leq}b{\leq}1$, the maximum sum DoF achieved by TRS designed by unicasting private messages with power $P^a$ is
\begin{IEEEeqnarray}{rcl}
d_{s{,}TRS}^{\max}(\mathcal{K}{,}\mathbf{r}{=}\mathbf{a})&{=}& \frac{K}{3}{+}\frac{K}{6}b{+}\frac{K}{2}a{+}\frac{b{-}a}{2}\epsilon{,}\label{eq:ds_scn1}
\end{IEEEeqnarray}
where $\epsilon$ is defined in \eqref{eq:ds2_scn1} and $\mathbf{r}{=}\mathbf{a}$ means that $r_k{=}a{,}{\forall}k{\in}\mathcal{K}$.
\end{myprop}

Obviously, the sum DoF achieved by the proposed TRS scheme strongly depends on $\epsilon$, i.e., the number of chains with odd number of vertices. Since there are at least two elements in a chain, the shortest length of a chain with odd number of vertices is $3$. Hence, the maximal value of $\epsilon$ is $\epsilon^*{=}\frac{K}{3}$, $\frac{K{-}2}{3}$ and $\frac{K{-}4}{3}$ when ${K}\bmod{3}{=}0$, ${K}\bmod{3}{=}2$ and ${K}\bmod{3}{=}1$, respectively. This indicates that the best topology that yields the greatest sum DoF has the property that in the generated graph there exist $\epsilon^*$ chains with three vertices and $\frac{K{-}3\epsilon^*}{2}$ chains with two vertices. Then, by substituting $\epsilon^*$ into \eqref{eq:ds_scn1}, we find that the best sum DoF is $\frac{K}{3}(1{+}b{+}a){-}\frac{k_m}{3}(b{-}a)$, where $k_m{\triangleq}2K\bmod{3}$. Besides, the worst topology has the property that all the chains have even number of vertices. The worst sum DoF is equal to $\frac{K}{3}(1{+}\frac{b}{2}{+}\frac{3a}{2})$.

\subsection{Discussions}

\begin{figure}[t]
\renewcommand{\captionfont}{\small}
\captionstyle{center}
\centering
\subfigure[All users are active.]{
                \centering
                \includegraphics[width=0.42\textwidth,height=2.75cm]{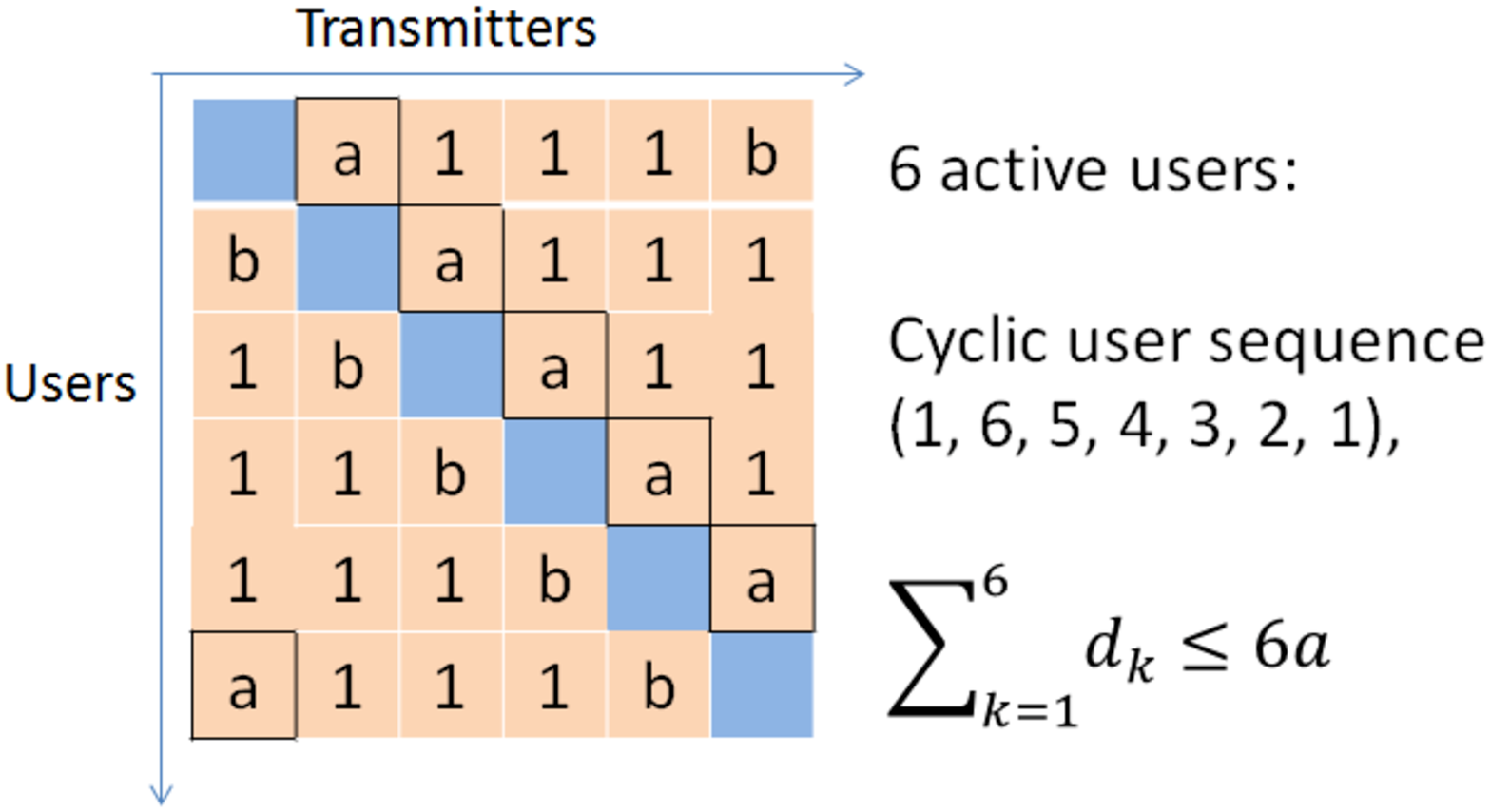}
                \label{fig:6active}
        }
        \subfigure[$5$ active users, $k{=}1{,}2{,}3{,}4{,}5$.]{
                \centering
                \includegraphics[width=0.42\textwidth,height=2.75cm]{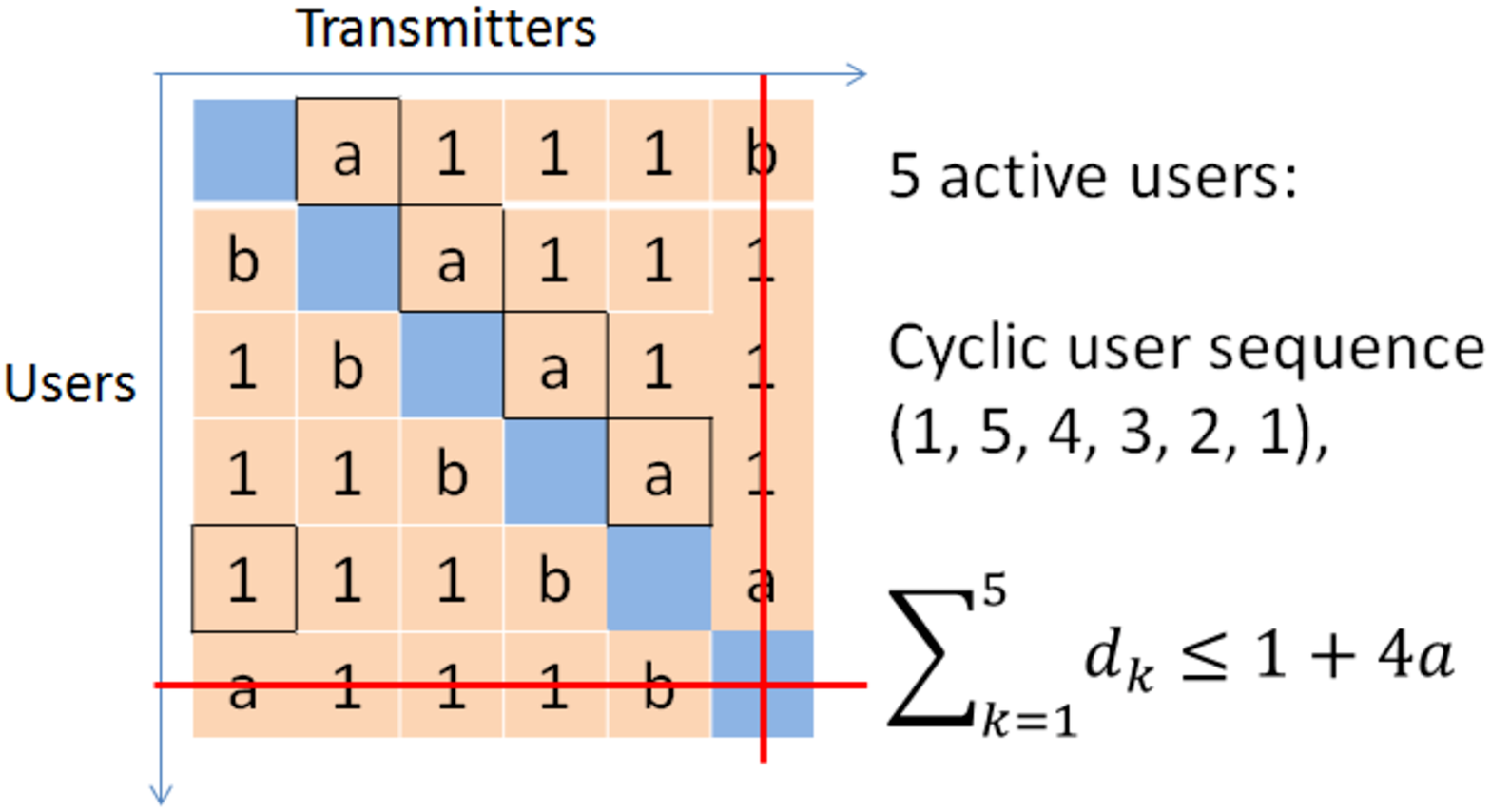}
                \label{fig:5active}
        }\\
        \subfigure[$4$ active users, $k{=}1{,}2{,}3{,}5$.]{
                \centering
                \includegraphics[width=0.42\textwidth,height=2.75cm]{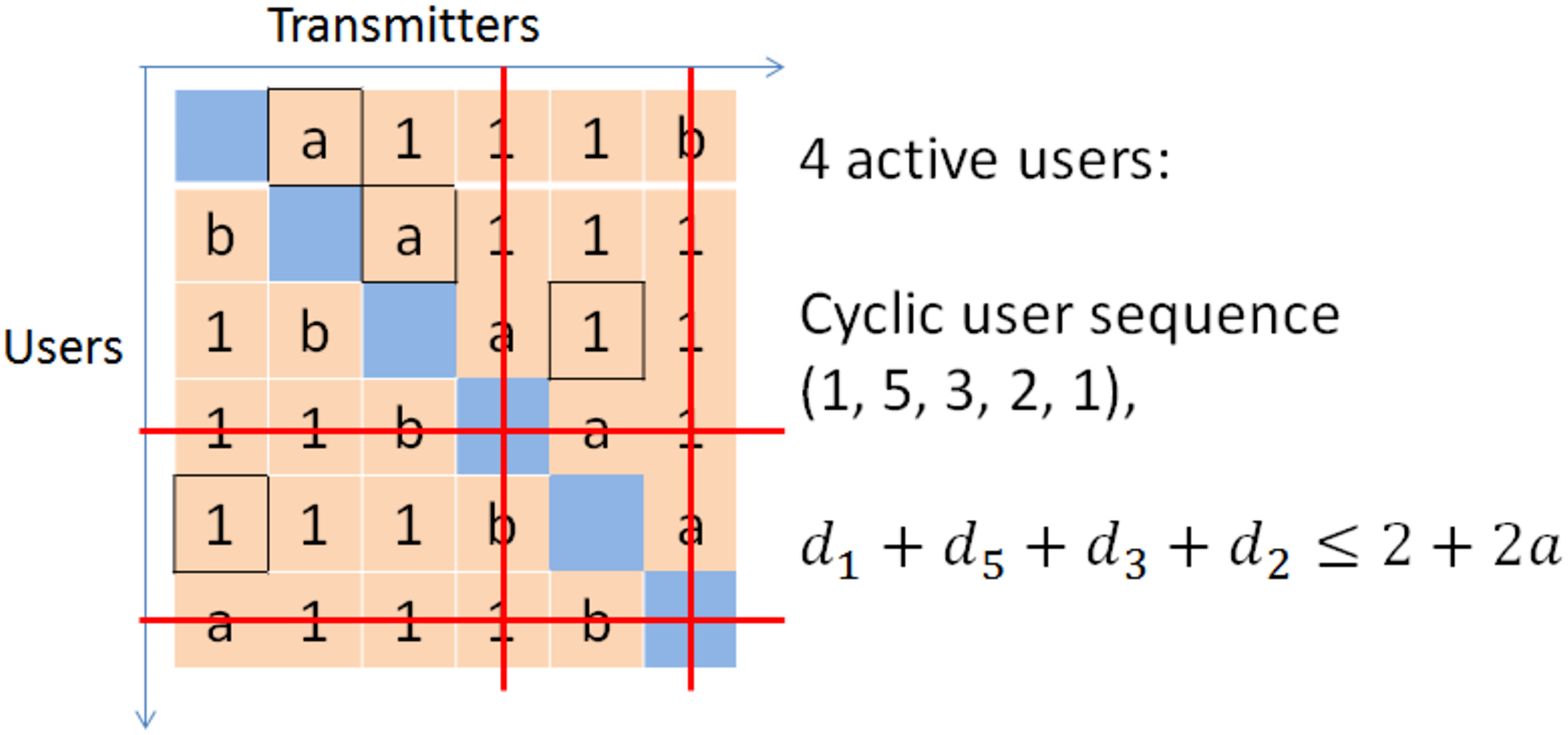}
                \label{fig:4active}
        }
        \subfigure[$3$ active users, $k{=}1{,}3{,}5$.]{
                \centering
                \includegraphics[width=0.42\textwidth,height=2.75cm]{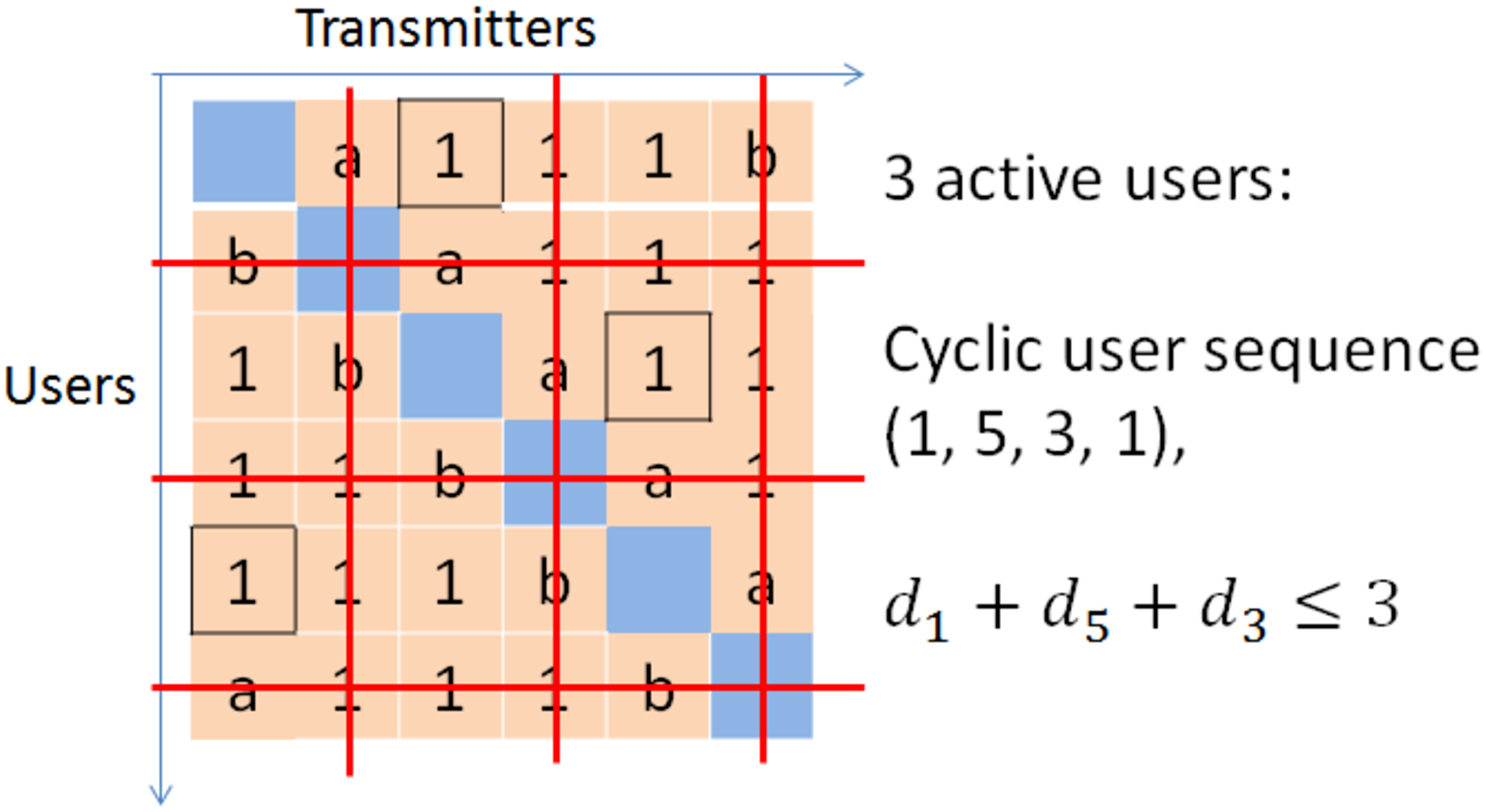}
                \label{fig:3active}
        }
\caption{Illustration of computing the sum DoF achieved by ZFBF with power control}\label{fig:zfbf}
\end{figure}
In this part, we compare the sum DoF achieved by the proposed TRS scheme with the sum DoF achieved by preliminary schemes. In RS, each user employs a fraction of the total power to unicast the private message, while employs the remaining power to transmit the common message. However, unlike the received signal presented in \eqref{eq:ykRS}, in the considered realistic scenarios, each user only decodes three common message transmitted by the dominant transmitters rather than all common messages. This fact implies that the DoF achieved by the common messages are specified by $|\mathcal{S}|$ different inequalities, rather than the single inequality $\sum_{k{\in}\mathcal{S}}d_k^c{\leq}1{-}\max_{j{\in}\mathcal{S}}r_j$ in \eqref{eq:DcpRS}. Hence, the achievable DoF region specified in Proposition \ref{prop:DRS} cannot be used to evaluate the DoF region achieved by RS in the considered realistic scenarios. Instead, we look at the sum DoF achieved by ZFBF with power control.

According to Remark \ref{rmk:ZFBF_region}, finding the maximum sum DoF achieved by ZFBF with power control requires a huge amount of efforts of evaluating the DoF region obtained for all the possible active user set $\mathcal{U}{\subseteq}\mathcal{K}$. To find a tractable result, we focus on the cyclic CSIT quality topology, i.e., $a_{k{,}k{+}1}{=}a$ and $a_{k{,}k{-}1}{=}b$, ${\forall}k{\in}\mathcal{K}$, where the index $k$ is based on modulus $K$. Note that according to the previous analysis, this topology is one of the worst topologies for the proposed TRS and the sum DoF achieved by the proposed TRS is $d_{s{,}TRS}^{\max}(\mathcal{K}{,}\mathbf{a}){=}\frac{K}{3}(1{+}\frac{b}{2}{+}\frac{3a}{2})$.

We evaluate the achievable sum DoF for each possible active user set $\mathcal{U}$ using \eqref{eq:DRS_KUp}, and the maximum of them yields the sum DoF achieved by ZFBF with power control. A $6$-user example is shown in Figure \ref{fig:zfbf}. Note that the $1$ in row $k$ and column $j$ with $j{\neq}k{-}1{,}k{+}1$ is obtained because the DoF achieved by ZFBF in the considered scenario is identical to the DoF achieved by ZFBF in a fully connected MISO IC with CSIT quality $a_{kj}{=}1$, ${\forall}j{\in}\mathcal{K}{\setminus}\{k{,}k{-}1{,}k{+}1\}$. When all the users are active, the achievable sum DoF $6a$ is given by inequality $\sum_{l{=}1}^{m}d_{i_l}^p{\leq}\sum_{l{=}1}^{m}a_{i_{l{-}1}i_l}$, where the cyclic user sequence is $(i_1{,}{\cdots}i_6){=}(1{,}6{,}5{,}4{,}3{,}2)$. Similarly, when there are $5$ active users, the achievable sum DoF is $1{+}4a$, given by the cyclic user sequence $(1{,}5{,}4{,}3{,}2)$. When there are $4$ active users, the best active users set that yields the maximum sum DoF is $\{1{,}2{,}3{,}5\}$ because user $5$ is not interfered with the other three users. The achievable sum DoF is $2{+}2a$. When there are $3$ active users, the sum DoF $3$ can be achieved by simply scheduling user $1$, $3$ and $5$. Hence, the maximum sum DoF achieved by ZFBF with power control is $\max\{6a{,}1{+}4a{,}2{+}2a{,}3\}$.

In general, when $K$ is an even number, by applying the same method as above, the maximum sum DoF achieved by ZFBF with power control is $\max_{n{:}{\frac{K}{2}}{\leq}n{\leq}K}\{K{-}n{+}(2n{-}K)a\}$, where $n$ stands for the number of active users. Thus, if $a{\geq}\frac{1}{2}$, we have $d_{s{,}zfbf}{=}Ka$; otherwise, we have $d_{s{,}zfbf}{=}\frac{K}{2}$. Comparing this result with the sum DoF $d_{s{,}TRS}^{\max}(\mathcal{K}{,}\mathbf{a}){=}\frac{K}{3}(1{+}\frac{b}{2}{+}\frac{3a}{2})$ achieved by the proposed TRS, it can be verified that $d_{s{,}TRS}^{\max}(\mathcal{K}{,}\mathbf{a}){>}d_{s{,}zfbf}$ as long as $b{+}3a{>}1$.

When $K$ is an odd number, the maximum sum DoF is achieved by ZFBF with power control is $\max\{\lfloor\frac{K}{2}\rfloor{,}\lfloor\frac{K}{2}\rfloor{-}1{+}a{+}b{,} 2\lfloor\frac{K}{2}\rfloor{-}n{+}1{+}(2n{-}2\lfloor\frac{K}{2}\rfloor{-}1)a\}$, where $\lfloor\frac{K}{2}\rfloor{+}2{\leq}n{\leq}K$. The number $\lfloor\frac{K}{2}\rfloor{-}1{+}a{+}b$ is the achievable sum DoF when there are $\lfloor\frac{K}{2}\rfloor{+}1$ active users. It is achieved by scheduling user $1$, $3$, $\cdots$, $2\lfloor\frac{K}{2}\rfloor{-}3$ who are not interfered with each other, and scheduling another two adjacent users, i.e., user $2\lfloor\frac{K}{2}\rfloor{-}1$ and user $2\lfloor\frac{K}{2}\rfloor$. The quantity $2\lfloor\frac{K}{2}\rfloor{-}n{+}1{+}(2n{-}2\lfloor\frac{K}{2}\rfloor{-}1)a$ is the achievable sum DoF when there are $n$ active users. It is achieved by scheduling $2\lfloor\frac{K}{2}\rfloor{-}n$ separated users, and $2n{-}2\lfloor\frac{K}{2}\rfloor$ adjacent users. Through some calculations, it can be verified that when $a{\geq}\frac{1}{2}$, we have $d_{s{,}zfbf}{=}Ka$; when $1{-}b{\leq}a{\leq}\frac{1}{2}$, we have $d_{s{,}zfbf}{=}\lfloor\frac{K}{2}\rfloor{-}1{+}a{+}b$; when $a{\leq}1{-}b$ and $a{\leq}\frac{1}{2}$, we have $d_{s{,}zfbf}{=}\lfloor\frac{K}{2}\rfloor$. Comparing with the sum DoF achieved by TRS, we conclude that $d_{s{,}TRS}^{\max}(\mathcal{K}{,}\mathbf{a}){>}d_{s{,}zfbf}$ as long as $b{+}3a{>}\frac{6}{K}\lfloor\frac{K}{2}\rfloor{-}2$.

Remarkably, we clarify that in the cyclic CSIT quality topology, the condition $b{+}3a{>}\frac{6}{K}\lfloor\frac{K}{2}\rfloor{-}2$ is a sufficient condition that TRS yields a sum DoF strictly greater than ZFBF with power control. When this condition does not hold, we can seek for an optimal active user set $\mathcal{S}^*$ and optimal power allocation policy $\mathbf{r}^*$, which maximize the sum DoF achieved by TRS.

\section{Conclusion}\label{sec:conclusion}
This paper, for the first time to our knowledge, studies the DoF of a $K$-cell interference channel where the CSIT of each interference link has an arbitrary quality of imperfectness. We firstly consider a Rate-Splitting approach where each user's data is split into a common part and a private part. The private messages are unicast along ZF-precoders using a fraction of total power, while the common messages are multicast using the remaining power and are to be decoded by all users. With an arbitrary power allocation for the private messages, we characterize the DoF region achieved by RS, and show that it covers the DoF region achieved by ZFBF with power control. Secondly, we propose a novel scheme called Topological RS. Compared to RS, the novelty lies in splitting the power used to transmit common messages into multiple layers. In each layer, with the properly assigned power level and ZF-precoders, we transmit common messages to be decoded by groups of users rather than all users. This multi-layer structure reduces the number of common messages decoded by each user, thus enhancing the DoF achieved by the common messages. The DoF region achieved by TRS is derived and is shown as a superset of the DoF region achieved by RS and ZFBF with power control. Besides, the sum DoF is studied from a graph theory perspective and the sum DoF of a class of realistic scenarios is characterized.

Apart from that, we would like to emphasize the usefulness of the weighted-sum interpretation that is used to design the TRS scheme. It bridges the MISO IC with imperfect CSIT and partially connected networks. By doing so, graph theory methodologies are introduced as powerful tools to analyze the DoF performance. From a sum DoF aspect, the benefit of weighted-sum interpretation is highlighted by deciding the DoF per common message or how many common messages can be transmitted in the corresponding partially connected network. This weighted-sum interpretation can be applied to many other scenarios, such as MISO networks with alternating CSIT qualities.

So far, the optimal DoF region and/or sum DoF of a $K$-cell interference channel with imperfect CSIT remains an open problem due to the lack of tight outer-bound. Our proposed TRS drives the inner-bound one-step further, and the obtained insights are transferrable to practical deployments.

%

\section*{Appendix}
\subsection{Proof of Proposition \ref{prop:DRS}}
The key part of the proof is to show the DoF region achieved by scheduling a subset $\mathcal{S}$ of users. Letting $\mathcal{D}_{RS}(\mathcal{S})$ denote the DoF region achieved by scheduling a subset $\mathcal{S}$ of users, it is obtained by taking the union of the DoF region achieved with all the possible power allocation $\mathbf{r}$, i.e., $\mathcal{D}_{RS}(\mathcal{S}){\triangleq}\bigcup_{{\forall}\mathbf{r}}\mathcal{D}_{RS}(\mathcal{S}{,}\mathbf{r})$.

With the proof presented in Appendix B, $\mathcal{D}_{RS}(\mathcal{S})$ is given by
\begin{IEEEeqnarray}{rcl}
\mathcal{D}_{RS}(\mathcal{S})&{=}&\bigcup_{{\forall}\mathcal{U}{\in}\mathcal{S}}\mathcal{D}_{RS}(\mathcal{S}{,}\mathcal{U}){,}\label{eq:DRS}
\end{IEEEeqnarray}
where $\mathcal{D}_{RS}(\mathcal{S}{,}\mathcal{U})$ is the set of $(d_1{,}\cdots{,}d_K){=}(d_1^c{,}\cdots{,}d_K^c){+}(d_1^p{,}\cdots{,}d_K^p)$ such that
\begin{IEEEeqnarray}{ccl}
d_k^p{=}0{,}{\forall}k{\in}\mathcal{K}{\setminus}\mathcal{U}{;}\,
0{\leq}d_k^p{\leq}1{,}{\forall}k{\in}\mathcal{U}{;}\,\sum_{l{=}1}^{m}d_{i_l}^p{\leq}\sum_{l{=}1}^{m}a_{i_{l{-}1}i_l}{,}
{\forall}(i_1{,}\cdots{,}i_m){\in}\Pi_{\mathcal{U}}{;}\label{eq:DRS_SUp}\\
d_k^c{=}0{,}{\forall}k{\in}\mathcal{K}{\setminus}\mathcal{S}{;}\,0{\leq}d_k^c{\leq}1{,}{\forall}k{\in}\mathcal{S}{;}\,
0{\leq}d_k^p{+}\sum_{j{\in}\mathcal{S}}d_j^c{\leq}1{,} {\forall}k{\in}\mathcal{U}{;}\nonumber\\
\sum_{j{\in}\mathcal{S}}d_j^c{+}\sum_{l{=}1}^{m}d_{i_l}^p{\leq}1{+}\sum_{l{=}2}^{m}a_{i_{l{-}1}i_l}{,} {\forall}(i_1{,}\cdots{,}i_m){\in}\Pi_{\mathcal{U}}{,}\label{eq:DRS_SUpc}
\end{IEEEeqnarray}
and $\Pi_{\mathcal{U}}$ is the set of all possible cyclic sequences of all subsets of $\mathcal{U}$ with cardinality no less than $2$.

In \eqref{eq:DRS_SUpc}, we see that for a certain set $\mathcal{U}$, by setting $d_k^c{=}0{,}{\forall}k{\in}\mathcal{K}{\setminus}\mathcal{S}$, $\mathcal{D}_{RS}(\mathcal{K}{,}\mathcal{U})$ becomes $\mathcal{D}_{RS}(\mathcal{S}{,}\mathcal{U})$. Then, it is immediate that $\mathcal{D}_{RS}(\mathcal{S}{,}\mathcal{U}){\subseteq} \mathcal{D}_{RS}(\mathcal{K}{,}\mathcal{U})$. This fact allows us to obtain the DoF region achieved by RS as $\mathcal{D}_{RS}{=}\bigcup_{{\forall}\mathcal{S}{\subseteq}\mathcal{K}{,}{\forall}\mathcal{U}{\subseteq}\mathcal{S}} \mathcal{D}_{RS}(\mathcal{S}{,}\mathcal{U}){=} \bigcup_{{\forall}\mathcal{U}{\subseteq}\mathcal{K}}\mathcal{D}_{RS}(\mathcal{K}{,}\mathcal{U})$, which completes the proof.

\subsection{Proof of \eqref{eq:DRS}}

The proof follows the footsteps in \cite[Section III.B and Appendix D]{Geng15TIN}. It has two steps. The first step is to characterize $\mathcal{D}_{RS}(\mathcal{S}{,}\mathcal{U})$. As it will be shown later on, the union of $\mathcal{D}_{RS}(\mathcal{S}{,}\mathcal{U})$ over $\mathcal{U}$ is a subset of $\mathcal{D}_{RS}(\mathcal{S})$. The second step is to show $\mathcal{D}_{RS}(\mathcal{S}){\subseteq}{\bigcup}_{{\forall}\mathcal{U}{\in}\mathcal{S}}\mathcal{D}_{RS}(\mathcal{S}{,}\mathcal{U})$.

\subsubsection{Step 1} For user $k{\in}\mathcal{S}{\setminus}\mathcal{U}$, we choose $r_k{=}0$ (Note that this choice is equivalent to $r_k{=}-\infty$ from a DoF perspective). Besides, we consider a polyhedral relaxation on the DoF tuple specified in \eqref{eq:DcpRS} by requiring $r_k{-}\max_{j{:}j{\in}\mathcal{S}{\setminus}k}(r_j{-}a_{kj})^+$ to be non-negative. Then, the achievable DoF region via polyhedral relaxation is the set of the DoF tuples such that
\begin{IEEEeqnarray}{rcl}
0{\leq}d_k^p{\leq}r_k{-}\max_{j{:}j{\in}\mathcal{U}{\setminus}k}(r_j{-}a_{kj})^+{,}{\forall}k{\in}\mathcal{U}{,} &\quad&
\sum_{k{\in}\mathcal{S}}d_k^c{\leq}1{-}\max_{j{\in}\mathcal{S}}r_j{.}\label{eq:dk_poly}
\end{IEEEeqnarray}
The polyhedral relaxation requires that the power exponents $\mathbf{r}$ such that the power of interference overheard by user $k$ is lower than the received power of user $k$'s desired private message. Otherwise, the power exponents $\mathbf{r}$ are regarded as achieving an invalid DoF tuple. However, according to \eqref{eq:DcpRS}, those power exponents actually lead to a valid DoF tuple. Hence, the DoF region is shrinked by the polyhedral relaxation. Now, denoting $d^c{=}\sum_{k{\in}\mathcal{S}}d_k^c$, we rewrite \eqref{eq:dk_poly} as
\begin{IEEEeqnarray}{rcl}
d_k^p&{\leq}&r_k{-}(r_j{-}a_{kj}){\Rightarrow}r_j{-}r_k{\leq}a_{kj}{-}d_k^p{,}{\forall}k{\in}\mathcal{U}{,}{\forall}j{\in}\mathcal{U}{\setminus}k{,} \label{eq:lkj}\\
d_k^p&{\leq}&r_k{\Rightarrow}-r_k{\leq}-d_k^p{,}\label{eq:l0k}\\
d_k^p&{\geq}&0{,}\\
d^c&{\leq}&1{-}r_k{\Rightarrow}r_k{\leq}1{-}d^c{,}{\forall}k{\in}\mathcal{U}{.}\label{eq:lk0}
\end{IEEEeqnarray}

Following the footsteps in \cite[Section III.B]{Geng15TIN}, we define a fully connected directed graph $\mathcal{G}(\mathcal{V}{,}\mathcal{E})$, where $\mathcal{V}{=}\{v_0{,}v_1{,}\cdots{,}v_{|\mathcal{U}|}\}$ is the vertex set and $\mathcal{E}$ is the set of the arcs. The length assigned to the arc from $v_j$ to $v_k$ is $l(v_j{,}v_k){=}a_{kj}{-}d_k^p$ for $i{,}j{\neq}0$, and the length assigned to the arc from $v_k$ to $v_0$ is $l(v_k{,}v_0){=}1{-}d^c$, while the length assigned to the arc from $v_0$ to $v_k$ is $l(v_0{,}v_k){=}-d_k^p$.

As defined in \cite{Comb}, a function $f$ is called a potential if for every two vertices, $a$ and $b$, such that $l(a{,}b){\geq}f(a){-}f(b)$ holds. Then, by setting $f(v_0){=}0$ and $f(v_k){=}r_k$, we see that any achievable DoF tuple such that \eqref{eq:dk_poly} holds, corresponds to a potential function for the directed graph. Moreover, the potential theorem \cite[Theorem 8.2]{Comb} suggests that there exists a potential function for a directed graph if and only if each circuit of $\mathcal{G}$ has a non-negative length. Thus, a DoF tuple is said satisfying \eqref{eq:dk_poly} if and only if each circuit of $\mathcal{G}$ has a non-negative length.
\begin{itemize}
  \item For the circuits $(v_0{,}v_k{,}v_0)$, we have $1{-}d^c{-}d_k^p{\geq}0$, yielding $d^c{+}d_k^p{\leq}1$, ${\forall}k{\in}\mathcal{U}$.
  \item For the circuits $(v_{i_0}{,}\cdots{,}v_{i_m})$ with $i_0{=}i_m$, ${\forall}(i_1{,}\cdots{,}i_m){\in}\Pi_{\mathcal{U}}$, ${\forall}m{\geq}2$, we have $\sum_{l{=}1}^{m}d_{i_l}^p{\leq}\sum_{l{=}1}^{m}a_{i_{l{-}1}i_l}$.
  \item For the circuits $(v_0{,}v_{i_1}{,}\cdots{,}v_{i_m}{,}v_0)$, we have $d^c{+}\sum_{l{=}1}^{m}d_{i_l}^p{\leq}1{+}\sum_{l{=}2}^{m}a_{i_{l{-}1}i_l}$.
\end{itemize}
Consequently, $\mathcal{D}_{RS}(\mathcal{S}{,}\mathcal{U})$ characterized by \eqref{eq:DRS_SUp} and \eqref{eq:DRS_SUpc} is immediate.

\subsubsection{Step 2}
To show $\mathcal{D}_{RS}(\mathcal{S}){\subseteq}{\bigcup}_{{\forall}\mathcal{U}{\in}\mathcal{S}}\mathcal{D}_{RS}(\mathcal{S}{,}\mathcal{U})$, we firstly introduce $\mathcal{D}_{RS}^\prime(\mathcal{S}{,}\mathcal{U})$ as
\begin{IEEEeqnarray}{rcl}
\mathcal{D}_{RS}^\prime(\mathcal{S}{,}\mathcal{U})&{=}& \{(d_1^c{,}\cdots{,}d_K^c{,}d_1^p{,}\cdots{,}d_K^p){\in}\mathcal{D}_{RS}(\mathcal{S}{,}\mathcal{U}){,}d_k^p{>}0{,}{\forall}k{\in}\mathcal{U}\}.
\end{IEEEeqnarray}
Then, it is clear that $\mathcal{D}_{RS}^\prime(\mathcal{S}{,}\mathcal{U}){\subseteq}\mathcal{D}_{RS}(\mathcal{S}{,}\mathcal{U})$, the remaining work is to show $\mathcal{D}_{RS}(\mathcal{S}){\subseteq}{\bigcup}_{{\forall}\mathcal{U}{\in}\mathcal{S}}\mathcal{D}_{RS}^\prime(\mathcal{S}{,}\mathcal{U})$. We aim to show that a DoF tuple lying outside ${\bigcup}_{{\forall}\mathcal{U}{\in}\mathcal{S}}\mathcal{D}_{RS}^\prime(\mathcal{S}{,}\mathcal{U})$ also lies outside $\mathcal{D}_{RS}(\mathcal{S})$. Such a DoF tuple has at least one of the following features:
\begin{itemize}
  \item $d_k^p{<}0$ or $d_k^p{>}1$ or $d^c{+}d_k^p{>}1$ for some user $k{\in}\mathcal{S}$.
  \item $\sum_{l{=}1}^{m}d_{i_l}^p{>}\sum_{l{=}1}^{m}a_{i_{l{-}1}i_l}$ for some cyclic sequence ${\forall}(i_1{,}\cdots{,}i_m){\in}\Pi_{\mathcal{U}}$.
  \item $d^c{+}\sum_{l{=}1}^{m}d_{i_l}^p{>}1{+}\sum_{l{=}2}^{m}a_{i_{l{-}1}i_l}$ for some users ${\forall}(i_1{,}\cdots{,}i_m){\in}\Pi_{\mathcal{U}}$.
\end{itemize}

It has been shown in \cite{Geng15TIN} that the DoF tuple satisfying the first and second feature cannot be included in $\mathcal{D}_{RS}(\mathcal{S})$. It remains to show that the DoF tuple satisfying the third feature cannot belong to $\mathcal{D}_{RS}(\mathcal{S})$. To this end, we employ the similar method in \cite{Geng15TIN}. Assuming the DoF tuple satisfying the third feature lies in $\mathcal{D}_{RS}(\mathcal{S})$. Then, there exists some $r_{i_l}$'s such that
\begin{IEEEeqnarray}{rcl}
d^c{+}\sum_{l{=}1}^{m}r_{i_l}{-}\max_{i_j{\in}\mathcal{S}{\setminus}i_l}(r_{i_j}{-}a_{i_li_j})^+&{>}&1{+}\sum_{l{=}2}^{m}a_{i_{l{-}1}i_l}\\
{\Rightarrow}d^c{-}1{+}\sum_{l{=}1}^{m}r_{i_l}{-}\max_{i_j{\in}\mathcal{S}{\setminus}i_l}(r_{i_j}{-}a_{i_li_j})^+{-}
\sum_{l{=}2}^{m}a_{i_{l{-}1}i_l}&{>}&0{.}
\label{eq:contradict}
\end{IEEEeqnarray}
Since $\max_{i_j{\in}\mathcal{S}{\setminus}i_l}(r_{i_j}{-}a_{i_li_j})^+{\geq}r_{i_{l{-}1}}{-}a_{i_li_{l{-}1}}$, the l.h.s. of \eqref{eq:contradict} can be upper-bounded as
\begin{IEEEeqnarray}{rcl}
d^c{-}1&{+}&r_{i_m}{-}\max_{i_j{\in}\mathcal{S}{\setminus}i_m}(r_{i_j}{-}a_{i_mi_j})^+{+}
\sum_{l{=}2}^{m}r_{i_{l{-}1}}{-}\max_{i_j{\in}\mathcal{S}{\setminus}i_{l{-}1}}(r_{i_j}{-}a_{i_{l{-}1}i_j})^+{-}a_{i_{l{-}1}i_l}\nonumber\\
&{\leq}&d^c{-}1{+}r_{i_m}{-}\max_{i_j{\in}\mathcal{S}{\setminus}i_m}(r_{i_j}{-}a_{i_mi_j})^+{+} \sum_{l{=}2}^{m}r_{i_{l{-}1}}{-}r_{i_{l{-}1}}{+}a_{i_{l{-}1}i_l}{-}a_{i_{l{-}1}i_l}{\leq}0{,}
\end{IEEEeqnarray}
which contradicts \eqref{eq:contradict}. This implies that $\mathcal{D}_{RS}(\mathcal{S}){\subseteq}{\bigcup}_{{\forall}\mathcal{U}{\in}\mathcal{S}}\mathcal{D}_{RS}^\prime(\mathcal{S}{,}\mathcal{U})$, which completes the proof.

\subsection{Proof of Proposition \ref{prop:DoFtuple}}
We firstly show the DoF tuple achieved by common message $\{w_k^i\}_{k{\in}\mathcal{S}}{,}i{\geq}2$ in \eqref{eq:DTRS_2} and \eqref{eq:DTRS_3}, and secondly show the DoF tuple achieved by private messages $\{w_k^1\}_{k{\in}\mathcal{S}}$ in \eqref{eq:DTRS_1}.

For user $k{,}{\forall}k{\in}\mathcal{K}$, when common messages of set $\mathcal{T}_k^i(\mathcal{S}{,}\mathbf{r})$ are decoded, it is assumed the common messages of set $\mathcal{T}_k^l(\mathcal{S}{,}\mathbf{r}){,}{\forall}l{>}i$, have been successively recovered and removed. Then, denoting the noise plus the interferences within the noise power by $\tilde{n}_k$, the received signal is expressed as
\begin{IEEEeqnarray}{rcl}
\tilde{y}_k&{=}&\sum_{j{:}w_j^i{\in}\mathcal{T}_k^i(\mathcal{S}{,}\mathbf{r})} \underbrace{\mathbf{h}_{kj}^H\mathbf{p}_j^iw_j^i}_{P^{a_{\pi(i{-}1)}}}{+}
\sum_{l{=}2}^{i{-}1}\sum_{j{:}w_j^i{\in}\mathcal{T}_k^i(\mathcal{S}{,}\mathbf{r})} \underbrace{\mathbf{h}_{kj}^H\mathbf{p}_j^iw_j^i}_{P^{a_{\pi(l{-}1)}}}{+}
\left(\underbrace{\mathbf{h}_{kk}^H\mathbf{p}_k^1w_k^1}_{P^{r_k}}{+}
\sum_{j{\in}\mathcal{S}{\setminus}k} \underbrace{\mathbf{h}_{kj}^H\mathbf{p}_j^1w_j^1}_{P^{r_j{-}a_{kj}}}\right){+}\underbrace{\tilde{n}_k}_{P^0}{.}
\end{IEEEeqnarray}
This system corresponds to a multiple-access-channel (MAC) where user $k$ wishes to decode messages of set $\mathcal{T}_k^i(\mathcal{S}{,}\mathbf{r})$. Following the capacity region of MAC \cite{BrunoBook}, the sum rate of any non-empty subset $\mathcal{M}{\subseteq}\mathcal{T}_k^i(\mathcal{S}{,}\mathbf{r})$ of messages are given by
\begin{IEEEeqnarray}{rcl}
\sum_{j{:}w_j^i{\in}\mathcal{M}}R_j&{\leq}&I(\mathcal{M}{;}\tilde{y}_k{|}\mathcal{T}_k^i(\mathcal{S}{,}\mathbf{r}){\setminus}\mathcal{M})
{=}h(\tilde{y}_k{|}\mathcal{T}_k^i(\mathcal{S}{,}\mathbf{r}){\setminus}\mathcal{M}){-}h(\tilde{y}_k{|}\mathcal{T}_k^i(\mathcal{S}{,}\mathbf{r})){,}
\label{eq:entropy_wki}
\end{IEEEeqnarray}
Considering that the input are random Gaussian codes, the entropies in \eqref{eq:entropy_wki} are equal to
\begin{IEEEeqnarray}{rcl}
h(\tilde{y}_k{|}\mathcal{T}_k^i(\mathcal{S}{,}\mathbf{r}){\setminus}\mathcal{M})&{=}& a_{\pi(i{-}1)}\log_2P{+}O(1){,}i{\geq}2{,}\label{eq:entropy_value_wki1}\\
h(\tilde{y}_k{|}\mathcal{T}_k^i(\mathcal{S}{,}\mathbf{r}))&{=}&a_{\pi(i{-}2)}\log_2P{+}O(1){,}i{\geq}3{,}\label{eq:entropy_value_wki2}\\
h(\tilde{y}_k{|}\mathcal{T}_k^2(\mathcal{S}{,}\mathbf{r}))&{=}& \max\{r_k{,}\max_{j{\in}\mathcal{S}{\setminus}k}r_j{-}a_{kj}\}\log_2P{+}O(1){,}\label{eq:entropy_value_wki3}
\end{IEEEeqnarray}
where $O(1)$ refers to the terms that do not change with $P$. Substituting \eqref{eq:entropy_value_wki1}, \eqref{eq:entropy_value_wki2} and \eqref{eq:entropy_value_wki3} into \eqref{eq:entropy_wki} and dividing them by $\log_2P$ lead to \eqref{eq:DTRS_2} and \eqref{eq:DTRS_3}.

When user $k$ decodes private message $w_k^1$, all the common messages have been recovered and removed. By treating the undesired private messages as noise, the rate of $w_k^1$ writes as
\begin{IEEEeqnarray}{rcl}
R_k^1&{\leq}&I(w_k^1{;}y_k{|}\mathcal{T}_k^i(\mathcal{S}{,}\mathbf{r}){,}i{=}2{,}\cdots{,}L{+}2)\\
&{=}&h(y_k{|}\mathcal{T}_k^i(\mathcal{S}{,}\mathbf{r}){,}i{=}2{,}\cdots{,}L{+}2){-}
h(y_k{|}\mathcal{T}_k^i(\mathcal{S}{,}\mathbf{r}){,}i{=}2{,}\cdots{,}L{+}2{,}w_k^1)\\
&{=}&h(\sum_{j{\in}\mathcal{S}}\mathbf{h}_{kj}^H\mathbf{p}_j^1w_j^1{+}\tilde{n}_k){-}
h(\sum_{j{\in}\mathcal{S}{\setminus}k}\mathbf{h}_{kj}^H\mathbf{p}_j^1w_j^1{+}\tilde{n}_k)\\
&{=}&r_k\log_2P{-}\max_{j{\in}\mathcal{S}{\setminus}k}(r_j{-}a_{kj})^+\log_2P{+}O(1){.}
\end{IEEEeqnarray}
Then, \eqref{eq:DTRS_1} is immediate.



\subsection{Proof of the sum DoF of the realistic scenario considered in Section \ref{sec:real}}
Without loss of generality, we consider the case $\Rowrk(\mathcal{M}^2){=}K{-}1$ and the edges $\mathcal{T}_k^2{=}\{\hat{w}_k^2{,}\hat{w}_{k{+}1}^2\}$, ${\forall}k{=}1{,}K{-}1$, and $\mathcal{T}_K{=}\mathcal{T}_{K{-}1}$. Clearly, in this scenario, there is one chain with length $K$. The inequality $\mathbf{M}^2\hat{\mathbf{d}}^i{\leq}1$ is explicitly expressed as $\hat{d}_1{+}\hat{d}_2{\leq}1$, $\hat{d}_2{+}\hat{d}_3{\leq}1$, $\hat{d}_3{+}\hat{d}_4{\leq}1$, $\cdots$, $\hat{d}_{K{-}1}{+}\hat{d}_K{\leq}1$. Adding up the inequalities with odd index yields
\begin{IEEEeqnarray}{rcl}
\text{If $K$ is even, }&\quad&\sum_{l{=}1}^{\frac{K}{2}}\hat{d}_{2l{-}1}^2{+}\hat{d}_{2l}^2{=}
\sum_{k{=}1}^{K}\hat{d}_k^2{\leq}\frac{K}{2}{;}\label{eq:Keven}\\
\text{If $K$ is odd, }&\quad&\sum_{l{=}1}^{\frac{K{-}1}{2}}\hat{d}_{2l{-}1}^2{+}\hat{d}_{2l}^2{=}\sum_{k{=}1}^{K{-}1}\hat{d}_k^2{\leq}\frac{K{-}1}{2}
{.}\label{eq:Kodd}
\end{IEEEeqnarray}
Inequality \eqref{eq:Keven} provides an upper-bound on the sum DoF of common messages $\{\hat{w}_k^2\}_{k{\in}\mathcal{K}}$ when $K$ is an even number. The equality holds with $\hat{d}_1^2{=}\hat{d}_3^2{=}\cdots{=}\hat{d}_{K{-}1}^2{=}1$. When $K$ is an odd number, we obtain an upper-bound on the sum DoF of common messages $\{\hat{w}_k^2\}_{k{\in}\mathcal{K}}$ by adding $\hat{d}_K^2$ to both sides of \eqref{eq:Kodd} as
\begin{IEEEeqnarray}{rcl}
\sum_{k{=}1}^{K}\hat{d}_k^2&{\leq}&\frac{K{-}1}{2}{+}\hat{d}_k^2{\leq}\frac{K{+}1}{2}{.}\label{eq:Kodd2}
\end{IEEEeqnarray}
The inequality \eqref{eq:Kodd2} is obtained due to the fact that $\hat{d}_K^2{\leq}1$. Then, using \eqref{eq:Kodd2} we can obtain the maximum sum DoF of common messages $\{\hat{w}_k^2\}_{k{\in}\mathcal{K}}$ as $\frac{K{+}1}{2}$ by taking $\hat{d}_1^2{=}\hat{d}_3^2{=}\cdots{=}\hat{d}_{K{-}2}^2{=}\hat{d}_{K}^2{=}1$.

\bibliographystyle{IEEEtran}

\bibliography{TIM}

\end{document}